\newif\ifdraft
\draftfalse

\ifdraft
\documentclass[a4paper,USenglish,cleveref,autoref,draft]{lipics-v2021}
\else
\documentclass[a4paper,USenglish,cleveref,autoref]{lipics-v2021}
\fi

\nolinenumbers
\usepackage{float}
\usepackage{xcolor}
\usepackage{amsmath}
\usepackage{listings}
\usepackage{subcaption}
\usepackage{xspace}
\usepackage{commands}
\usepackage{enumitem}
\usepackage{thmtools}
\usepackage{thm-restate}

\def\naive{na\"{\i}ve}

\newcommand{\figref}[1]{\figurename~\ref{fig:#1}}
\newcommand{\secref}[1]{Sec.~\ref{sec:#1}}
\newcommand{\appref}[1]{App.~\ref{appendix:#1}}
\newcommand{\defref}[1]{Def.~\ref{defn:#1}}
\newcommand{\defcaseref}[2]{Def.~\ref{defn:#1}(\ref{defn:#1:#2})}
\newcommand{\exref}[1]{Example~\ref{ex:#1}}

\newcommand{\thmref}[1]{Theorem~\ref{thm:#1}}

\newcommand{\problem}[1]{Challenge~#1}

\newcommand\noeditcolour{} 

\usepackage[normalem]{ulem} 
\usepackage{todonotes}
\usepackage{xspace}
\newcommand{\myedit}[2]{\ifdefined\noeditcolour{#2}\else{\color{#1}{#2}}\fi}
\DeclareRobustCommand{\myout}[2]{\ifdefined\noeditcolour{}\else\myedit{#1}{\texorpdfstring{\sout{#2}}{#2}}\fi}
\newcommand{\myfootnote}[3]{\ifdefined\noeditcolour{}\else\myedit{#2}{\footnote{\myedit{#2}{#1:~ #3}}}\fi}
\newcommand{\mytodo}[3]{\mytodogeneral{#1}{#2}{#3}{}}

\newcommand{\mytodogeneral}[4]{\ifdefined\noeditcolour{}\else\todo[#4color=#2,author=\color{white}{#1}]{\color{white}{#3}}\xspace\fi}

\def\ascolour{orange!70!black}
\def\as{\myedit{\ascolour}}
\def\asout{\myout{\ascolour}}
\def\asfootnote{\myfootnote{AS}{\ascolour}}
\def\astodo{\mytodo{AS}{\ascolour}}

\def\zgcolour{violet!80}
\def\zg{\myedit{\zgcolour}}
\def\zgout{\myout{\zgcolour}}

\newcommand{\zgedit}[2]{\zgout{#1 }{\zg{#2}}}

\def\nvcolour{blue}
\def\nv{\myedit{\nvcolour}}
\def\nvout{\myout{\nvcolour}}

\def\nvtodo{\mytodo{NV}{\nvcolour}}

\makeatletter
\def\operator#1{\@ifnextchar\bgroup{\operatorarg{\ensuremath{#1}}}{\ensuremath{#1}}}
\def\operatorarg#1#2{\@ifnextchar\bgroup{\operatorargmore{#1}{#2}}{\operatorargtypeset{#1}{#2}}}
\def\operatorargmore#1#2#3{\operatorarg{#1}{#2,#3}}
\def\operatorargtypeset#1#2{{#1}{\ensuremath{(#2)}\xspace}}
\def\sqoperator#1{\@ifnextchar\bgroup{\sqoperatorarg{\ensuremath{#1}}}{\ensuremath{#1}}}
\def\sqoperatorarg#1#2{\ensuremath{#1{[}#2{]}}}

\def\fixedoperator#1{\@ifnextchar\bgroup {\fixedoperatorarg{#1}}{\ensuremath{#1}}}
\def\fixedoperatorarg#1#2{\fixedoperatorparse{#1}#2~}
\def\fixedoperatorparse#1#2,#3~{\ensuremath{{#2}{.}{#1}{(#3)}}}

\def\infixoperator#1{\@ifnextchar\bgroup {\infixoperatorarg{#1}}{\ensuremath{#1}}}
\def\infixoperatorarg#1#2{\infixoperatorparse{#1}#2~}
\def\infixoperatorparse#1#2#3~{\ensuremath{{#2}~{#1}~{#3}}}
\makeatother

\def\sat{\operator{\satc}}
\def\satc{\textit{sat}}
\def\ocrefine{\operator{\ocrefinec}}
\def\ocrefinec{\textit{refine}}

\newskip \point \point =1pt
\setbox134=\hbox{\leavevmode\raise0\point\hbox{${\langle}\kern-2.5\point{\langle}$}}
\setbox135=\hbox{\raise0\point\hbox{${ \rangle}\kern-2.5\point{ \rangle}$}}

\setbox136=\hbox{\leavevmode\raise0\point\hbox{${\lfloor}\kern-3.0\point{\lfloor}$}}
\setbox137=\hbox{\leavevmode\raise0\point\hbox{${\rfloor}\kern-2.7\point{\rfloor}$}} 
\newcommand{\floorBrackets}[2]{\copy136{#1}\copy137_{#2}}
\setbox138=\hbox{\leavevmode\raise0\point\hbox{${[}\kern-1.5\point{[}$}}
\setbox139=\hbox{\raise0\point\hbox{${]}\kern-1.5\point{]}$}}

\usepackage{fancybox}

\newcommand{\callout}[2]{
\begin{center}
  \shadowbox{
    \begin{minipage}{0.9\textwidth}
      \textbf{#1}\ #2
    \end{minipage}
  }
\end{center}}

\newcommand{\union}{\cup}
\newcommand{\intersect}{\cap}

\newcommand{\funorder}{\ensuremath{>}}
\newcommand{\ts}{\textit{ts}}

\newcommand{\corders}{\ensuremath{O}}
\newcommand{\genordername}{\ensuremath{\Gamma}}
\newcommand{\corderstop}{\genordername}
\newcommand{\rwpath}[2]{\ensuremath{#1_1 \rightarrow \ldots \rightarrow #1_#2}}
\newcommand{\gtpath}[2]{\ensuremath{#1_1 >_{\trms} \ldots >_{\trms} #1_#2}}
\newcommand\grange{\ensuremath{\corderstop}\xspace}

\usepackage[ruled,vlined]{algorithm2e}
\usepackage{wrapfig}
\usepackage{listings}
\usepackage{pifont}

\usepackage[scaled=0.8]{beramono}  

\def\withcolor{}

\ifdefined\withcolor
  \definecolor{fstarblue}{rgb}{0.0, 0.0, 1.0}
  \definecolor{haskellstr}{rgb}{0.2, 0.2, 0.6}
  \definecolor{haskellred}{rgb}{1.0, 0.0, 0.0}
  \definecolor{gray_ulisses}{gray}{0.55}
  \definecolor{castanho_ulisses}{rgb}{0.59,0.42,0.15}
  \definecolor{preto_ulisses}{rgb}{0.55,0.28,0.59}
  \definecolor{green_ulises}{rgb}{0.59,0.42,0.15}
\else
	\definecolor{fstarblue}{gray}{0.1}
	\definecolor{haskellstr}{gray}{0.1}
	\definecolor{haskellred}{gray}{0.1}
	\definecolor{gray_ulisses}{gray}{0.1}
	\definecolor{castanho_ulisses}{gray}{0.1}
	\definecolor{preto_ulisses}{gray}{0.1}
	\definecolor{green_ulisses}{gray}{0.1}
\fi

\def\codesize{\small}

\lstdefinelanguage{HaskellUlisses} {
	basicstyle=\ttfamily\codesize,
	sensitive=true,
	belowskip={0.1em},
	aboveskip={0.1em},
    morecomment=[s]{(*}{*)},
	morecomment=[l][\color{gray_ulisses}\ttfamily\itshape\codesize]{--},
	morestring=[b]",
	stringstyle=\color{haskellstr},
	basewidth={0.2em},
	showstringspaces=false,
	numberstyle=\codesize,
	numberblanklines=true,
	showspaces=false,
	breaklines=true,
	showtabs=false,
	tabsize=4,
    literate={ 
             {->}{{$\rightarrow$}}2
			 {<=>}{{$\Leftrightarrow$}}1
             {~int}{{$\mathbb{Z}$}}1
             {~nat}{{$\mathbb{N}$}}1
			 {==>}{{$\Longrightarrow$}}1
			 {`feq`}{{$\eqinfix$}}1
			 {ka}{{k${}_a$}}1
			 {kb}{{k${}_b$}}1
			 {dollar}{{$\$$}}1
			 {dsl}{{d$_{sl}$}}2
			 {dfs}{{d$_{fs}$}}2
			 {rsl}{{r$_{sl}$}}2
			 {rfs}{{r$_{fs}$}}2
			 {dlm}{{d$_{lm}$}}2
           },
	emph=
	{[1] Tot, Type, bool, Lemma, ensures, requires, Ifc, IFC, IfcClearance, GlobalInt, GTot
	},
	emphstyle={[1]\color{fstarblue}},
	emph=
	{[2] class, match, with, if, then, else, let, rec, type, val, in, instance, data, measure, where, effect,noeq, private
	},
	emphstyle={[2]\color{castanho_ulisses}},
	emph=
	{[3]
        lattice, value, equals, canFlow, meet, join, bottom, top, 
        lawBot, lawFlowReflexivity, lawFlowAntisymetry, lawFlowTransitivity, 
        lawMeet, lawJoin, labels, 
        lt, lmeet, ljoin, lcanFlow, eq,
        labeled, labeledTCB
	},
	emphstyle={[3]\color{preto_ulisses}\textbf},
	emph=
	{[4]
        Low, Medium, High
	},
	emphstyle={[4]\color{green_ulises}\textbf},
	emph=
	{[5] assume, admit, admitP
	},
	emphstyle=[5]\color{red}\textbf,
	emph={[6] leq, equals, join', c_0, c_1
	},
	emphstyle=[6]\color{green}\textbf,
}

\lstnewenvironment{code}
{\lstset{language=HaskellUlisses}}
{}

\lstnewenvironment{scode}
{\lstset{language=HaskellUlisses,basicstyle=\ttfamily\footnotesize,keepspaces,mathescape}}
{}

\lstnewenvironment{mcode}
{\lstset{language=HaskellUlisses,columns=fullflexible,keepspaces,mathescape}}
{}

\lstnewenvironment{ccode}
{\lstset{language=C,columns=fullflexible,keepspaces,mathescape}}
{}

\lstMakeShortInline[language=HaskellUlisses,mathescape,keepspaces,mathescape,basicstyle=\ttfamily\codesize,breakatwhitespace]!

\usepackage{multirow}

\newtheorem{invariant}{\algoname Invariant}

\title{REST: Integrating Term Rewriting with Program Verification}

\author{Zachary {Grannan}}{University of British Columbia, Vancouver, Canada}{zgrannan@cs.ubc.ca}{https://orcid.org/0000-0000-0000-0000}{}
\author{Niki {Vazou}}{IMDEA Software Institute, Madrid, Spain}{niki.vazou@imdea.org}{https://orcid.org/0000-0000-0000-0000}{}
\author{Eva {Darulova}\thanks{This work was partly done while the author was at MPI-SWS}}{Uppsala University, Uppsala, Sweden}{eva.darulova@it.uu.se}{https://orcid.org/0000-0000-0000-0000}{}
\author{Alexander J.~{Summers}}{University of British Columbia, Vancouver, Canada}{alex.summers@ubc.ca}{https://orcid.org/0000-0000-0000-0000}{}
\authorrunning{Z. Grannan, N. Vazou, E. Darulova, and A. J. Summers}
\Copyright{Zachary Grannan, Niki Vazou, Eva Darulova, and Alexander J. Summers}
\begin{document}


\maketitle
\begin{abstract}
We introduce \acronym, a novel term rewriting technique for theorem proving that
uses online termination checking and can be integrated with existing program
verifiers.
\acronym{} enables flexible but terminating term rewriting for theorem proving by:
(1) exploiting newly-introduced term orderings that are more permissive than
standard rewrite simplification orderings;
(2) dynamically and iteratively selecting orderings based on the path of rewrites taken so far; and
(3) integrating external oracles that allow steps that cannot be justified with rewrite rules.
\as{Our \algoname approach is designed around an easily implementable core algorithm, parameterizable by choices of term orderings and their implementations; in this way our approach can be easily integrated into existing tools.} We implemented \algoname as a Haskell library and
incorporated it into Liquid Haskell's evaluation strategy, \asout{thus }extending Liquid Haskell with rewriting rules.
We evaluated our \acronym implementation
by comparing it against both existing rewriting techniques and E-matching
and by showing that it can be used to supplant manual lemma application in many existing Liquid Haskell proofs.
\end{abstract}

\section{Introduction}
\label{sec:introduction}

For all disjoint sets $s_0$ and $s_1$, the identity $(s_0 \union s_1) \intersect s_0 = s_0$ can be proven in many ways.
Informally accepting this property is easy,
but a machine-checked formal proof may require the instantiation of multiple set
theoretic axioms. Analogously, further proofs relying on this identity may
themselves need to apply it as a previously-proven lemma.
For example, proving functional correctness
of any program that relies on a set data structure typically requires the instantiation of set-related lemmas.
Manual instantiation of such universally quantified equalities is tedious, and
the burden becomes substantial for more complex proofs: a proof author needs to identify exactly which equalities to instantiate and with which arguments; in the context of program verification, a wide variety of such lemmas are typically available.
Given this need, most
program verifiers provide some automated technique or heuristics for instantiating universally quantified equalities.

For the wide range of practical program verifiers that are built upon SMT
solvers (\eg{}
\cite{dafny,FilliatreP13,LiquidHaskell,MuellerSchwerhoffSummers16,FStar,framac}),
quantified equalities can naturally be expressed in the SMT solver's logic.
However, relying solely on such solvers' E-matching techniques \cite{Detlefs05}
for quantifier instantiation (as the majority of these verifiers do) can lead to
both non-termination and incompletenesses that may be unpredictable
\cite{leino_trigger_2016} and challenging to diagnose
\cite{BeckerMuellerSummers19}. \as{The theory of how to prove that an E-matching-based encoding of equality reasoning guarantees termination and completeness is difficult and relatively unexplored \cite{dross2016adding}.}

A classical alternative approach to automating equality reasoning is \emph{term
rewriting}~\cite{Huet77}, which can be used to encode
lemma properties as (directed) rewrite rules, matching terms against the
existing set of rules to identify potential rewrites; the termination
of these systems is a well-studied problem~\cite{dershowitz1987termination}.
Although SMT solvers often perform rewriting as an internal simplification step,
verifiers built on top typically cannot access or customize these rules, \eg{}
to add previously-proved lemmas as rewrite rules. By contrast, many mainstream
proof assistants (\eg \Coq\cite{Coq}, \Isabelle\cite{Isabelle},
\Lean\cite{Lean}) 
provide automated, customizable term rewriting tactics.
However, the rewriting functionalities of
mainstream proof assistants either do not ensure the termination of rewriting
(potentially resulting in divergence, for example Isabelle) or enforce termination checks that
are overly restrictive in general, potentially rejecting necessary rewrite
steps (for example, Lean).

In this paper, we present \emph{\acronym (REwriting and Selecting Termination
orderings)}: a novel technique that equips program verifiers with automatic
lemma application facilities via term rewriting, enabling equational reasoning
with complementary strengths to E-matching-based techniques. While term
rewriting in general does not guarantee termination, our technique weaves together
three key technical ingredients to automatically generate and explore
guaranteed-terminating restrictions of a given rewriting system while typically
retaining the rewrites needed in practice:
(1) \algoname compares terms using well-quasi-orderings derived from (strict)
simplification orderings; thereby facilitating common and important rules such as
commutativity and associativity properties.
(2) \algoname simultaneously considers an entire family of term orderings; selecting the
appropriate term ordering to justify rewrite steps \emph{during term rewriting
itself}.
(3) \algoname allows integration of an \emph{external oracle} that generates additional
steps outside of the term rewriting system. This allows the incorporation of
reasoning steps awkward or impossible to justify via rewriting rules, all
without compromising the termination and relative completeness guarantees of our
overall technique.

\subparagraph*{Contributions and Overview} We make the following contributions:
\begin{enumerate}
\item We design and present a new approach (\acronym) for applying term rewriting rules and simultaneously selecting appropriate term orderings to permit as many rewriting steps as possible while guaranteeing termination (\secref{approach}).
\item We introduce ordering constraint algebras, an abstraction for reasoning
        effectively about multiple (and possibly infinitely
        many) term orderings simultaneously (\secref{orderings}).
\item We introduce and formalize recursive path quasi-orderings (\rpqo{}s)
derived from the well-known recursive path ordering
\cite{dershowitz1982orderings} (\secref{orderings:def}). RPQOs are more
permissive than classical RPOs, and so let us prove more properties.
\item We formalize and prove key results for our technique: soundness, relative completeness, and termination (\secref{metaprop}).
\item We implement \algoname as a stand-alone library, and integrate the
        \algoname library into Liquid Haskell to facilitate automatic lemma instantiation
(\secref{implementation}).
\item We evaluate \acronym{} by comparing it to other term
rewriting tactics and E-matching-based axiomatization, and show that it can substantially simplify equational reasoning proofs (\secref{evaluation}).
\end{enumerate}
We discuss related work in \secref{related}; we begin (\secref{overview}) by identifying \numchallenges key problems that all need solving for a reliable and automatic integration of term rewriting into a program verification tool.
%

\section{\Numchallenges Challenges for Automating Term Rewriting}
\label{sec:challenges}\label{sec:overview}
In this section, we describe \emph{\numchallenges key challenges} that naturally
arise when term rewriting is used for program verification and outline
how \algoname is designed to address them.
To illustrate the challenges, we use
simple verification goals that involve uninterpreted functions and the
set operators ($\emptyset$, $\union$, $\intersect$)
that satisfy the standard properties of \figref{identities}.
The variables $x,y,z$ are implicitly
quantified\footnote{over sets; we omit explicit types in such formulas, whose
type-checking is standard.} in these rules. In formalizations of set theory,
such properties may be assumed as (quantified) axioms, or proven as
lemmas and then used in future proofs.

\begin{figure}[t]
\begin{center}
$$
 \begin{array}{|l|rcl|}
    \hline
    \textit{Name} & \multicolumn{3}{c|}{\textit{Formula}}  \\
  \hline
\textit{idem-union} & X \union X &=& X \\
\textit{idem-inter} & X \intersect X &=& X \\
\textit{empty-union} & X \union \emptyset &=& X  \\
\textit{empty-inter} & X \intersect \emptyset &=& \emptyset \\
\textit{commut-union}  & X \union Y &=& Y \union X \\
\textit{symm-inter} & X \intersect Y &=& Y \intersect X\\
\textit{distrib-union} &(X \union Y) \intersect Z &=& (X \intersect Z) \union (Y \intersect Z) \\
\textit{distrib-inter} &(X \intersect Y) \union Z &=& (X \union Z) \intersect (Y \union Z) \\
\textit{assoc-union} &X \union (Y \union Z) &=& (X \union Y) \union Z \\
  \hline
  \end{array}
$$
\end{center}
\caption{Set identities used for examples in this section. Variables $X,Y,Z$ are implicitly quantified. We write the
    binary functions $\union, \intersect$ infix; along with (nullary)
    $\emptyset$ these are fixed function symbols.}\label{fig:identities}
\end{figure}

Term rewriting systems (defined formally in \secref{background}) are a standard
approach for formally expressing and applying equational reasoning (rewriting
terms via known identities). A term rewriting system consists of a finite set of
\emph{rewrite rules}, each consisting of a pair of a \emph{source term} and a
\emph{target term}, representing that terms matching a rule's source can be
replaced by corresponding terms matching its target. For example, the rewrite rule
$X \union \emptyset \rightarrow X$ can replace set unions of some set $X$
and the empty set with the corresponding set $X$. Rewrite rules are applied to a
term $t$ by identifying some subterm of $t$ which is equal to a rule's source
under some substitution of the source's free variables (here, $X$, but not
constants such as $\emptyset$); the subterm is then replaced with the
correspondingly substituted target term. This rewriting step \emph{induces an
equality} between the original and new terms. For instance, the example rewrite
rule above can be used to rewrite a term $f(s_0 \union \emptyset)$ into
$f(s_0)$, inducing an equality between the two.

Rewrite rules classically come with two restrictions: the free variables of the
target must all occur in the source and the source must not be a single variable. This
precludes rewrite rules which invent terms, such as
$\emptyset \rightarrow X \intersect \emptyset$, and those that trivially lead to
infinite derivations. Under these restrictions, the first four identities induce
rewrite rules from left-to-right (which we denote by \eg{} \textit{idem-inter$\rightarrow$}), while the remaining induce rewrite rules in
both directions (\eg{} \textit{assoc-union$\rightarrow$} vs.~\textit{assoc-union$\leftarrow$}).

Next, we present a simple proof obligation taken from \cite{leino2013verified} in the style of equational reasoning (\emph{calculational proofs}) supported in the Dafny program verifier \cite{dafny}.
\begin{example}
\label{ex:identities}
We aim to prove, for two sets $s_0$ and $s_1$ and some unary function $f$ on sets, that, if the sets are disjoint (that is, $s_1 \intersect s_0 = \emptyset$), then $f((s_0 \union s_1) \intersect s_0) = f(s_0)$.
\[
\begin{array}{rrcll}
\textit{Equational Proof:}\quad&f((s_0 \union s_1) \intersect s_0) & = & f((s_0 \intersect s_0) \union (s_1 \intersect s_0)) \quad& \textit{(distrib-union}{\rightarrow}\textit{)}\\
&& = & f(s_0 \union (s_1 \intersect s_0)) \quad& \textit{(idem-inter}{\rightarrow}\textit{)}\\
&& = & f(s_0 \union \emptyset) \quad& \textit{(disjointness ass.}{\rightarrow}\textit{)}\\
&& = & f(s_0) \quad& \textit{(empty-union}{\rightarrow}\textit{)}\\
\multicolumn{5}{l}{\color{gray}{\as{\text{(Possible Term Ordering, as explained shortly: RPO instance with }\intersect \funorder \union\text{)}}}}
\end{array}
\]

\end{example}

This manual proof closely follows the user annotations employed in the corresponding Dafny proof \cite{leino2013verified}; the application of the function $f$ serves only to illustrate equational reasoning on subterms. Every step of the proof could be explained by term rewriting, hinting at the possibility of an \emph{automated} proof in which term rewriting is used to solve such proof obligations. In particular, taking the term rewriting system naturally induced by the set identities of \figref{identities} \emph{along with} the assumed equality expressing disjointness of $s_0$ and $s_1$ results in a term rewriting system in which the four proof steps are all valid rewriting steps.

In the remainder of the section, we consider what it would take to make term
rewriting effective for reliably automating such verification tasks. Perhaps unsurprisingly, there
are multiple problems with the simplistic approach outlined so far. The first
and most serious is that term rewriting systems in general \emph{do not
guarantee termination}; a proof search may continue
indefinitely by repeatedly applying rewrite rules. For example,
the rules $\textit{distrib-union}$ and $\textit{distrib-inter}$ can lead to an
infinite derivation
$ (s_{0} \cup s_{1}) \cap s_{2} \rightarrow
  (s_{0} \cap s_{2}) \cup (s_{1} \cap s_{2}) \rightarrow
  (s_{0} \cup (s_{1} \cap s_{2})) \cap (s_{2} \cup (s_{1} \cap s_{2})) \rightarrow \ldots
$

\callout{Challenge 1:}{Unrestricted term rewriting systems do not guarantee termination.}

To ensure termination (as proved in~\thmref{terminating}) \algoname follows
the classical approach \as{of restricting a term-rewriting system to a variant in which sequences of term rewrites (\emph{rewrite paths}) are allowed only if each consecutive pair of terms is \emph{ordered} according to some term ordering which rules out infinite paths.}\asout{
of terminating  term rewriting systems and requires
that rewrite applications decrease the size of the term with respect to
a well-founded order on terms.}

\as{For example,} \emph{Recursive path orderings} (RPOs)~\cite{dershowitz1982orderings}
define well-founded orders $>_{\trms}$ on terms $\trms$
based on an underlying well-founded strict partial order $\funorder$ on
\emph{function symbols}. Intuitively, such orderings use $\funorder$ to order terms with
different top-level function symbols, combined with the properties of a
\emph{simplification order} \cite{dershowitz_simplification_1979} (\eg{}
compatibility with the subterm relation). \as{Different choices of the underlying $>$ parameter yield different RPO instances that order different pairs of terms; in particular, potentially allowing or disallowing certain rewrite paths.}

In~\exref{identities}, an RPO based on a partial order where
$\intersect \funorder \union$ and $\intersect \funorder \emptyset$
permits all the rewriting steps, that is,
the left-hand-side of each equation is greater than the right-hand-side.

Sadly, this ordering will not permit the rewriting
steps required by our next example.

\begin{example}
\label{ex:identitiestwo}
We aim to prove, for two sets $s_0$ and $s_1$ and some unary function $f$ on sets, that, if $s_1$ is a subset of $s_0$ (that is, $s_0 \union s_1 = s_0$), then $f((s_0 \intersect s_1) \union s_0) = f(s_0)$.
\[
\begin{array}{rrcll}
\textit{Equational Proof:}\quad&f((s_0 \intersect s_1) \union s_0) & = & f((s_0 \union s_0) \intersect (s_1 \union s_0)) \quad& \textit{(distrib-inter}{\rightarrow}\textit{)}\\
&&  = & f(s_0 \intersect (s_1 \union s_0)) \quad& \textit{(idem-union}{\rightarrow}\textit{)}\\
&& = & f(s_0 \intersect (s_0 \union s_1)) \quad& \textit{(commut-union}{\rightarrow}\textit{)}\\
&& = & f(s_0 \intersect s_0) \quad& \textit{(subset ass.}{\rightarrow}\textit{)}\\
&& = & f(s_0) \quad& \textit{(idem-inter}{\rightarrow}\textit{)}\\
\multicolumn{5}{l}{\color{gray}{\as{\text{(Possible Term Ordering: RPQO instance, explained shortly, with }\intersect \funorder \union\text{)}}}}
\end{array}
\]
\end{example}
An \rpo based on an ordering where $\intersect \funorder \union$ (as required by
\exref{identities}) will not permit the first step of this proof (since the \rpo
ordering first compares the top level function symbols).
Instead,  this step requires an \rpo based on an ordering where $\union \funorder \intersect$.
To accept \textit{both} this proof step
\textit{and} the \exref{identities}
we need \asout{a rewriting system that permits rewrites}\as{different restrictions of the rewrite rules for different proofs; in particular, different rewrite paths may be}\asout{
that are} ordered according to \rpos that are based on
different function orderings.

To generalize this problem we will call \rpos a term ordering \textit{family}
that is \textit{parametric} with respect to the underlying function ordering.
Thus, a concrete \rpo term ordering \as{(called an \emph{instance} of the family)} is obtained after the parametric function
ordering is instantiated. With this terminology, the next challenge \as{can be stated as follows:}

\callout{Challenge \challengeDiffInstantiations:}{Different proofs require
   different term orderings within a family.}\label{challenge:diffinstantiations}
\as{Note that enumerating all term orderings in a term ordering family is typically impractical (this set is often very large and may be infinite).} To address this challenge,
\algoname uses a novel algebraic structure (\secref{oca})
to \as{allow for an abstract representation of sets of term orderings with which one can} efficiently check whether \as{any instance of a chosen}\asout{a} term ordering family \asout{can be instantiated
with a parameter such that the resulting ordering} can orient the necessary rewrite
steps to complete a proof.

Going back to \exref{identitiestwo},
the \rpo instance with $\union \funorder \intersect$ will permit all the steps,
apart from the commutativity axiom expressed by $\textit{(commut-union}{\rightarrow}\textit{)}$.
To permit this step we need an ordering for which
$t_1 \union t_2 >_{\trms} t_2 \union t_1$.
But for \rpo instances, as well as for many other term orderings,
the terms $t_1 \union t_2$ and $t_2 \union t_1$ are equivalent and thus cannot
be oriented; associativity axioms are also similarly challenging.
Since many proofs require such \as{properties, it is important in practice for rewriting to support them.}

\callout{Challenge \challengeWQO:}{Strict orderings restrict
commutativity and associativity steps.}
To address this challenge \algoname relaxes the strictness constraint \as{by
requiring the chosen term ordering family to consist (only) of \emph{thin
well-quasi-orderings}} (\defref{orderings:prop}). Intuitively, such orderings
permit rewriting to \as{terms which are \emph{equal} according to the ordering,
but such equivalence classes of terms must be finite. In \secref{orderings} we
show how to lift well-known families of term orderings to analogous and
more-permissive families of thin well-quasi-orders. In particular, we show how
to lift RPOs to a particularly powerful family of term orderings that we call
\emph{recursive path quasi-orderings (RPQOs)} (\defref{rpqo}), whose
instances allow us} to accept \exref{identitiestwo}.

\as{Despite the permissiveness of RPQOs}, there \as{remain some}\asout{are many} rewrite derivations that will be rejected by
all term orderings in the RPQO family.
For example, consider the following proof that set union is
monotonic with respect to the subset relation:
\begin{example}
\label{ex:setmono}
We aim to prove, for sets $s_0$, $s_1$, and $s_{2}$, that, if $s_1$ is a subset
of $s_0$ (that is, $s_0 \union s_1 = s_0$), then
$(s_{2} \cup s_{1}) \cup (s_{2} \cup s_{0}) = s_{2} \cup s_{0}$.
\[
\begin{array}{rrcll}
\textit{Equational Proof:}\quad&(s_{2} \cup s_{1}) \cup (s_{2} \cup s_{0}) & = & s_{2} \cup (s_{1} \cup (s_{2} \cup s_{0})) \quad& \textit{(assoc-union}{\leftarrow}\textit{)}\\
&& = & s_{2} \cup ((s_{1} \cup s_{2}) \cup s_{0}) \quad& \textit{(assoc-union}{\rightarrow}\textit{)}\\
&& = & s_{2} \cup ((s_{2} \cup s_{1}) \cup s_{0}) \quad& \textit{(commut-union}{\rightarrow}\textit{)}\\
&& = & s_{2} \cup (s_{2} \cup (s_{1} \cup s_{0})) \quad& \textit{(assoc-union}{\leftarrow}\textit{)}\\
&& = & s_{2} \cup (s_{2} \cup (s_{0} \cup s_{1})) \quad& \textit{(commut-union}{\rightarrow}\textit{)}\\
&& = & s_{2} \cup (s_{2} \cup s_{0}) \quad& \textit{(subset ass.}{\rightarrow}\textit{)}\\
&& = & (s_{2} \cup s_{2}) \cup s_{0} \quad& \textit{(assoc-union}{\rightarrow}\textit{)}\\
&& = & s_{2} \cup s_{0} \quad& \textit{(idem-union}{\rightarrow}\textit{)}\\
\multicolumn{5}{l}{\color{gray}{\as{\text{(Possible Term Ordering: any KBQO instance)}}}}
\end{array}
\]
\end{example}
The above rewrite rule steps cannot be oriented by any RPQO,
but are trivially oriented by a quasi-ordering that is based
on the syntactic size of the term,
\eg a quasi-ordering based on the \as{well-known} Knuth-Bendix family \as{of term orderings \cite{kbo}}.
Yet, a Knuth-Bendix quasi-ordering (KBQO\as{, defined in \secref{orderings}}) cannot be used on our previous two examples\as{; fixing even a single choice of term ordering \emph{family} would still be too restrictive in general}.

\callout{Challenge \challengeDiffOrders:}{Some proofs require different families of
  term orderings.}
To address this challenge, \algoname (\secref{algo2}) is defined parametrically
\as{in the choice and representation of a} term ordering family.

Finally, although equational reasoning is powerful enough for these examples,
general verification problems usually require reasoning beyond the scope of
simple rewriting. For example, simply altering \exref{identities} to express the
disjointness hypothesis instead via cardinality as $|s_0 \intersect s_1| = 0$
means that, to achieve a similar proof, reasoning within the theory of sets is
necessary to deduce that this hypothesis implies the equality needed for the
proof; this is beyond the abilities of term rewriting.
\callout{Challenge \challengeOracle:}{Program verification needs proof steps not expressible by rewriting.}
To address this challenge, \as{our} \algoname~\as{approach allows the integration of an external oracle that can generate equalities not justifiable by term rewriting, while still guaranteeing  termination}\asout{combines rewriting with an external oracle}
(\secref{oracle}).

\section{The \acronym{} Approach}\label{sec:algo}\label{sec:approach}
We develop \acronym{} to tackle the above \numchallenges challenges
and integrate a flexible, expressive, and guaranteed-terminating
term rewriting system with a verification tool.
%
\acronym consists of an interface for defining term orderings and an
algorithm for exploring the rewrite paths supported by the term orderings.
In \secref{termorder} we describe the representation of term orderings in
\algoname and how they address Challenges \challengeDiffInstantiations and \challengeDiffOrders.
In \secref{algo2} we describe the \algoname algorithm that is parametric to
these orderings  and
\secref{oracle} describes the integration with external oracles
(\problem{\challengeOracle}).

\subsection{Representation of Term Orderings in \algoname}\label{sec:termorder}

Rather than considering individual term orderings, \algoname operates on indexed
sets (families) of term orderings (whose instances must all be thin well-quasi-orderings [\defref{orderings:prop}]).

\begin{definition}[Term Ordering Family]\label{defn:tofamily}
  A \emph{term ordering family} $\corderstop$  is a set of \as{thin well-quasi-orderings} on terms, indexed by some
  parameters $P$. \as{An \emph{instance} of the family is a term ordering} obtained
  by a particular instantiation of $P$.
\end{definition}

For example, the recursive path ordering is defined parametrically with
respect to a precedence on function symbols, and therefore defines a term
ordering family \as{indexed by this choice of function symbol ordering}.

A core \as{concern}\asout{component} of \algoname is determining whether any \as{instance of a given term ordering family}\asout{ inside a
family} can orient a rewrite path. However, term ordering families cannot
directly compare terms; doing so requires choosing an ordering inside the
family. The root of \problem{\challengeDiffInstantiations} is that choosing an
ordering in advance is too restrictive: different orderings are necessary to
complete different proofs.
The idea behind \algoname's search algorithm is to
address this challenge by simultaneously considering all orderings in the family
when considering rewrite paths and continuing the path so long as it can be
oriented by \emph{any} ordering.

To demonstrate the technique, we show how \algoname's approach can be derived
from a \naive~algorithm. The purpose of the algorithm is to determine if any ordering in
a family $\genordername$ can orient a path $\rwpath{t}{n}$; \ie if there is a
$>_{\trms}~\in~\corderstop$ such that $\gtpath{t}{n}$.

\begin{figure}[h]
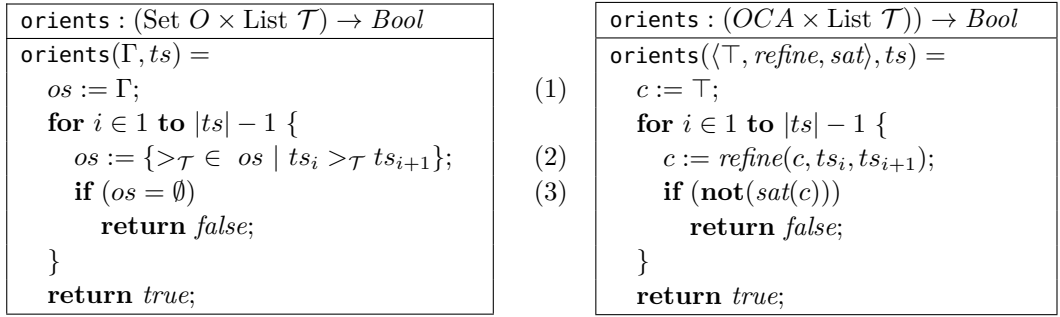

  \begin{minipage}{0.55\textwidth}
  $$
  \begin{array}{|lc|c}
    \cline{1-2}
    \texttt{orients}: ( \settype{\corders}\times \listtype{\trms}) \rightarrow \textit{Bool} & & \\
    \cline{1-2} 
  \texttt{orients}(\corderstop, ts) = & &  \\
  \quad os := \corderstop; & & \quad (1)\quad \\
  \quad \textbf{for}~i \in 1~\textbf {to}~|ts|-1~\{ & & \\
  \quad \quad os := \{>_{\trms}~\in~os~|~ts_{i} >_{\trms}
    ts_{i + 1} \}; & & \quad (2) \quad\\
  \quad \quad \textbf{if}~(os = \emptyset) &  & \quad (3) \quad \\
  \quad \quad \quad \textbf{return}~\textit{false}; & & \\
  \quad \} & & \\
  \quad \textbf{return}~\textit{true}; & &  \\
  \cline{1-2}
  \end{array}
  $$
\end{minipage}
\begin{minipage}{0.40\textwidth}
  $$
\begin{array}{|lc|}
  \hline
  \texttt{orients}: (\ocatype\times\listtype{\trms})) \rightarrow \textit{Bool} & \\
  \hline 
  \texttt{orients}(\langle \octop, \ocrefinec, \satc \rangle, ts) = &  \\
\quad c := \octop; & \\
\quad \textbf{for}~i \in 1~\textbf {to}~|ts|-1~\{ & \\
\quad \quad c := \ocrefinec(c, ts_{i}, ts_{i+1}); &  \\
\quad \quad \textbf{if}~(\textbf{not}(\satc(c))) & \\
\quad \quad \quad \textbf{return}~\textit{false}; & \\
\quad \} & \\
\quad \textbf{return}~\textit{true}; &  \\
\hline
\end{array}
$$
\end{minipage}
\caption{
  Two algorithms that determine if an ordering in
  the term ordering family $\corderstop$ can orient a path of terms $ts$.
  \textbf{Left} presents the \naive, exhaustive algorithm.
  \textbf{Right} is using the ordering
constraint algebra $\langle \octop, \ocrefinec, \satc \rangle$ that
returns true iff an ordering in $\corderstop$ can orient $ts$
without explicitly constructing any term orderings.
\corders is the type of a term ordering.}
\label{fig:naive-rewrite}
\label{fig:naive-rewrite-oca}
\end{figure}

\as{The \naive~algorithm is depicted on the left of \figref{naive-rewrite}.} The \naive~algorithm works
iteratively, computing the set of orderings $os$ that can orient an
increasingly-long \asout{prefix of the }path, short-circuiting if the set becomes empty.
\zgedit{Unfortunately, the algorithm is not practical as it}The algorithm \zgedit{requires enumerating}{enumerates} each ordering in
$\corderstop$ and \zgedit{comparing}{compares} terms with each ordering (potentially multiple times).
\as{\zgedit{This}{Unfortunately, this} enumeration is not practical:} some term ordering families \as{have infinite or prohibitively
large numbers of instances.} \algoname~avoids these issues by allowing the set of term orderings
to be abstracted via a structure called an \Ocalgebra (OCA, \defref{oca} of~\secref{oca}).

An OCA for a term ordering family $\genordername$ \as{consists of} a type $C$
\as{along with four parameters $\gamma : C \rightarrow  \mathcal{P}(\genordername)$}, $\octop : C$,
$\ocrefinec : C \rightarrow \trms \rightarrow \trms \rightarrow C$, and
$\satc : C \rightarrow \textit{Bool}$. \as{$C$ is a type whose elements \emph{represent} subsets of $\genordername$. The function $\gamma$ is the \emph{concretisation function} of the OCA, not needed programmatically but instead defining the \emph{meaning} of elements of $C$ in terms of the subsets of the term ordering family they represent. The remaining three functions} correspond to the operations on sets
of term orderings used in lines (1), (2), and (3) of the \naive~algorithm.
 \octop represents the set of all
term orderings in $\genordername$, $\ocrefinec(c, t, u)$ filters the set of orderings represented by $c$ to include only
those where $t >_{\trms} u$, and $\satc(c)$ is a predicate that returns true if
the set of orderings represented by $c$ is nonempty.
\figref{naive-rewrite-oca} on the right shows how the \oca can be used to perform an
equivalent computation to the \naive~algorithm, without explicitly instantiating
sets of term orderings.
The OCA plays a role similar to abstract interpretation in a program analysis, where $C$ is an abstraction over sets of
term orderings, and the results of the abstract
operations on $C$ correspond to their concrete equivalents.
\zg{Namely, we have
$\gamma(\octop) = \corderstop$,
$\gamma(\ocrefine{c}{\tt_l}{\tt_r}) = \{ \ord \;|\; \ord~\in \gamma(c) \;\land\; {\tt_l}  \ord  {\tt_r}\}$, and
$\satc(c) \;\Leftrightarrow\; \gamma(c)\neq\emptyset$.}

The ordering constraint algebra enables three main advantages compared to direct
computation with sets of term orderings:
\begin{enumerate}
      \item The number of term orderings can be very large, or even
        infinite, thus making enumeration of the entire set intractable.
      \item An OCA can provide efficient implementations for \ocrefinec~and \satc~by
        exploiting properties of the term ordering family. Comparing terms
        using the constituent term orderings requires repeating the comparison for each
        ordering, despite the fact that most orderings will differ in ways
        that are irrelevant for the comparison.
      \item The OCA does not impose any requirements on the type of $C$ or the
implementation of $\octop, \ocrefinec$, and $\satc$. For example, an OCA can use
$\octop$ and $\ocrefinec$ to construct logical formulas, with $\satc$ using an
external solver to check their satisfiability. Alternatively, it could define
$C$ to be sets of term orderings that are reasoned about explicitly, and
        implement $\octop, \ocrefinec$, and $\satc$ as the operations of the
        \naive~algorithm.
\end{enumerate}

We now describe how the \algoname algorithm uses the OCA to explore rewrite
paths.

\subsection{The \algoname Algorithm}\label{sec:algo2}

\begin{figure}[t]
  $$
  \begin{array}{|lcl|}
  \hline
  \algoname: (\ocatype\times \rreltype \times \trms\;\times\; (\trms \rightarrow \settype{\trms})) \rightarrow \settype{\trms} & & \\
  \hline 
  \algoname(\langle \octop, \ocrefinec, \satc \rangle,\rrel, \tt_0, \evalSymb) = & &  \\
  \quad o := \emptyset; & & \\
  \quad p := [([\tt_{0}], \octop)]; & & \\
  \quad \textbf{while } ( p \text{ is not empty} )\{& & \\
  \quad \quad \textbf{pop} (ts, c) \text{ from } p; & & \\
  \quad \quad \tt := \text{last } \ts; & & \\
  \quad \quad o := o \cup \{ \tt \}; & & \\
  \quad \quad \textbf{foreach } (\tt' \textit{such that } \tt' \not \in \ts \;\land\; (\goestor{\rrel
  }{\tt}{\tt'} \;\lor\; \tt'\in\evalSymb(\tt))) \{  & & \\
  \quad \quad \quad \textbf{if}\ (
                      \tt'\in\evalSymb(\tt)
                       \lor
                        (\goestor{\rrel}{\tt}{\tt'} \land \sat{\ocrefine{c}{\tt}{\tt'}})
  )\{ & & \\
  \quad\quad \quad \quad \textbf{push } (\ts \concat [\tt'],\ocrefine{c}{\tt}{\tt'}) \textbf{ to } p & & \\
  \quad\quad \quad \} & & \\
  \quad\quad \} & & \\
  \quad  \} & & \\
  \quad \textbf{return}\ o ; & & \\
  \hline
  \end{array}
  $$
\caption{The \algoname algorithm.}
\label{fig:algo-rewrite}
\end{figure}

\figref{algo-rewrite} presents the \algoname algorithm. The algorithm takes
four parameters.
The first parameter is an OCA $\langle \octop, \ocrefinec, \satc \rangle$, as
discussed above.
The algorithm's second parameter, $\rrel$, is a finite set of term rewriting rules
(not required to be terminating); for example, we could pass the oriented
rewrite rules corresponding to \figref{identities}.
The third parameter
$\tt_{0}$ is the term from which term rewrites are sought.
The final parameter \evalSymb{} acts as
an external oracle, generating additional rewrite steps that need \emph{not}
follow from the term rewriting rules $R$. To simplify the explanation, we will
initially assume that $\evalSymb = \lambda t.\emptyset$, \ie{} this parameter
has no effect. Our algorithm produces a set of terms, each of which are
reachable by \emph{some} rewrite path beginning from $\tt_{0}$, and for which
\emph{some} ordering allows the rewrite path. The algorithm addresses
\problem{\challengeTermination} (termination; \thmref{terminating}) because every path must be finite:
no ordering could orient an infinite path.

Our algorithm operates in worklist fashion, storing in $p$ a list of pairs $(\ts,\ordconstraints)$ where $\ts$ is a non-empty list of terms representing a rewrite path already explored (the head of which is always $\tt_{0}$) and $\ordconstraints$ tracks the ordering constraints of the path so far.
The set $o$ records the output terms (initially empty): all terms discovered (down any rewrite path) equal to $\tt_{0}$ via the rewriting paths explored.

While there are still rewrite paths to be extended, \ie $p$ is not empty, a tuple
$(ts, \ordconstraints)$ is popped from $p$.
\algoname puts \tt, \ie the last term of the path, into the set of output terms $o$ and considers all terms $\tt'$
that are: (a) not \emph{already} in the path and (b) reachable by a single rewrite step of $\rrel$ (or returned by the function $\evalSymb$ explained later).
The crucial decision of whether or not to extend a rewrite path with the
additional step $\tt \rightarrow \tt'$ is handled in the \textbf{if} check of \algoname.
This check is to guarantee termination, by enforcing that we only add rewrite steps which would leave the extended path still justifiable by \emph{some} term ordering, as enforced by the \sat~check.

\begin{figure}[t]
  \includegraphics[width=\textwidth]{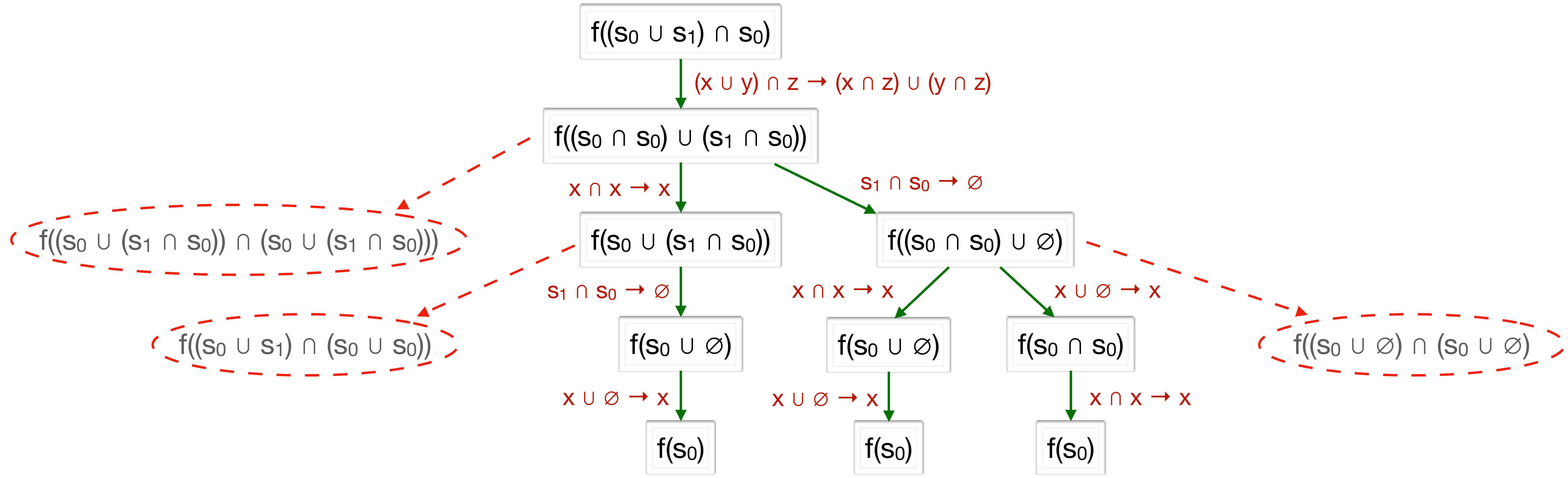}
  \caption{A visualization of \acronym{} running on the term from
    \exref{identities}. Each path through the tree shown represents a rewrite
    path uncovered by our algorithm; the edge labels show the rewrite rule
    applied. The red dotted lines indicate rewrite steps rejected by \algoname.}\label{fig:vis-example1}
\end{figure}

\figref{vis-example1} visualizes the rewrite paths explored by our algorithm for
a run corresponding to the problem from \exref{identities}, using the OCA for the
recursive path quasi-ordering (\secref{orderings:oc})\footnote{We omit the
commutativity rules from this run, just to keep the diagram easy to visualize,
but our implementation handles the example easily with or without them.}. The
manual proof in \exref{identities} corresponds to the right-most path in this
tree; the other paths apply the same reasoning steps in different orders. In our
implementation, we optimize the algorithm to avoid re-exploring the same term
multiple times unless this could lead to further rewrites being discovered
(\cf{} \secref{implementation}).

The arrow from the root of the tree to its child corresponds to the first
rewrite \algoname applies:
$f((s_0 \union s_1) \intersect s_0) \rightarrow f((s_0 \intersect s_0) \union (s_1 \intersect s_0))$.
This rewrite step can only be oriented by RPQOs with precedence
$\intersect > \union$; therefore applying this rewrite constrains the set of
RPQOs that \algoname must consider in subsequent applications. For example, the
rewrite to the left child of
$f((s_0 \intersect s_0) \union (s_1 \intersect s_0))$ can only be oriented by
RPQOs with precedence $\union > \intersect$. Since no RPQO can have both
$\intersect > \union$ and $\union > \intersect$, no RPQO can orient the entire
path from the root; \algoname must therefore reject the rewrite. On the other
hand, the rewrite to the right child can be oriented by any RPQO where
$s_{0} > \emptyset$, $s_{1} > \emptyset$, or $\intersect > \emptyset$. The path
from the root can thus continue down the right-hand side, as there are RPQOs that
satisfy both $\intersect > \union$ and the other conditions. The subsequent
rewrites down the right-hand side do not impose any new constraints on the
ordering:
$f((s_{0} \intersect s_{0}) \cup \emptyset) >_{\trms} f(s_{0} \intersect s_{0}) >_{\trms} f(s_{0})$
in all RPQOs.

Similarly, \algoname will prove \exref{identitiestwo} but will reject
\exref{setmono} when the input OCA represents RPQO orderings.
As shown in our benchmarks (Table \ref{table:evaluation:case-studies} of \secref{evaluation}),
\exref{setmono} is solved by \algoname with an OCA for the Knuth-Bendix term ordering family.

\subsection{Integrating an External Oracle}\label{sec:oracle}

Finally, to tackle \problem{\challengeOracle}, we turn to the (so far ignored) third parameter
of the algorithm, the external oracle \evalSymb{}. In the example variant presented at the end of \secref{challenges}, such a function might
supply the rewrite step $s_0 \intersect s_1 \rightarrow \emptyset$ by analysis
of the logical assumption $|s_0 \intersect s_1| = 0$, which goes beyond
term-rewriting. More generally, any external solver capable of producing rewrite
steps (equal terms) can be connected to our algorithm via \evalSymb{}. In our
implementation in Liquid Haskell, we use the
pre-existing \emph{Proof by Logical Evaluation (PLE)} technique \cite{PLE},
which complements rewriting with the expansion of program function definitions,
under certain checks made via SMT solving. Our only requirements on the oracle
$\evalSymb{}$ are that the binary relation on terms generated by calls to it is bounded (finitely-branching) and strongly normalizing (\cf{} \secref{metaprop}).

Our algorithm therefore flexibly allows the interleaving of term rewriting steps
and those justified by the external oracle; we avoid the potential for this
interaction to cause non-termination by conditioning any further rewriting steps
on the fact that the entire path (including the steps inserted by the oracle)
can be oriented by at least one candidate term ordering.

The combination of our interface for defining term orderings via \oca{}s, a
search algorithm that effectively explores all rewrites enabled by the
orderings, and the flexible possibility of combination with external solvers via
the oracle parameter makes \acronym{} very adaptable and powerful in practice.


\section{Well-Quasi-Orderings and the Ordering Constraint Algebra}\label{sec:term-order}\label{sec:orderings}

Term orderings are typically defined as \textit{strict well-founded} orderings; this
requirement ensures that rewriting will obtain a normal form. However, as
mentioned in \problem{\challengeWQO}, the restriction to strict orderings
limits what can be achieved with rewriting. In this section we
describe the derivation of well-quasi-orderings from strict orderings (\secref{wqo}) and
introduce
Knuth-Bendix quasi-orderings (\secref{kbqo}) and
recursive path quasi-orderings (\secref{orderings:def}), two novel term ordering families
respectively based on the classical recursive path and Knuth-Bendix orderings.
In addition, we formally introduce ordering constraint algebras (\secref{oca})
and use them to develop an efficient ordering constraint algebra for
RPQOs.

\subsection{Well-Quasi-Orderings}
\label{sec:wqo}

We define well-quasi-orderings in the standard way.

\begin{definition}[Well-Quasi-Orderings]\label{defn:orderings:prop}

A relation $\geqslant$ is a quasi-order if it is reflexive and transitive. Given
elements $t$ and $u$ in $S$, we say $t \approx u$ if $t \geqslant u$ and
$u \geqslant t$. A quasi-order $\geqslant$ is also characterized as:
\begin{enumerate}
\item \textit{\wqo}, when for all infinite chains $x_1, x_2, \ldots$ there exists an
$i, j, i < j$ such that $x_j \geqslant x_i$,
\item \textit{thin}, when forall $t \in S$, the
set $\{u \in S~|~t \approx u \}$ is finite, and
\item \textit{total}, when for all $t, u \in S$ either $t \geqslant s$ or $s \geqslant t$.
\end{enumerate}
\end{definition}

Well-quasi-orderings are not required to be antisymmetric,
however the corresponding strict part of the ordering must be well-founded.
Hence, a WQO derives a strict ordering over equivalence classes of terms;
\algoname also requires that these equivalence classes are finite (\ie the
ordering is thin). With this
requirement, \algoname guarantees termination by exploring only duplicate-free paths.

Many simplification orderings can be converted into more permissive \wqo{}s.
Intuitively, given an ordering $>_{o}$ its quasi-ordering derivation also
accepts equal terms, so we denote it as $\geqslant_{o}$. We next present two
such derivations.

\subsubsection{Knuth-Bendix Quasi-Orderings (KBQO)}
\label{sec:kbqo}

The Knuth-Bendix ordering~\cite{kbo} is a well-known simplification ordering used
in the Knuth-Bendix completion procedure.
Here, we present a simplified version of the ordering, used by \algoname
that is using ordering to only compare ground terms.

\begin{definition}
  A \textit{weight function} $w$ is a function $\ops \rightarrow \mathbb{N}$,
where $w(f) > 0$ for all nullary functions symbols, and $w(f) = 0$ for at most one unary
function symbol. $w$ is \textit{compatible} with a quasi-ordering $\geqslant_{\ops}$ on
$\ops$ if, for any unary function $f$ such that $w(f) > 0$, we have
$f >_{\ops} g$ for all $g$. $w(t)$ denotes the weight of a
term $t$, such that $w(f(t_{1}, \ldots, t_{n})) = w(f) + \sum\limits_{1 \leqslant i \leqslant n} w(t_{i})$
\end{definition}

\begin{definition}[Knuth-Bendix ordering (KBO) on ground terms]
The Knuth-Bendix Ordering $>_{kbo}$ for a given weight function $w$ and
compatible precedence order $\geqslant_{\ops}$ is defined as
$f(t_{1}, \ldots, t_{m}) = t >_{kbo} u = g(u_{1}, \ldots, u_{n})$ iff
$w(t) \geqslant w(u)$, and:
  \begin{enumerate}
    \item $w(t) > w(u)$, or
    \item $f >_{\ops} g$, or
    \item $f \geqslant_{\ops} g$, and
$(t_{1}, \ldots, t_{m}) >_{\textit{kbolex}} (u_{1}, \ldots, u_{n})$.
  \end{enumerate}
  Where $>_{\textit{kbolex}}$ performs a lexicographic comparison using $>_{kbo}$ as the underlying ordering.
\end{definition}

Intuitively, KBO compares terms by their weights, using $\geqslant_{\ops}$ and
the lexicographic comparison as ``tie-breakers'' for cases when terms have equal
weights. However, as $\geqslant$ is already a well-quasi-ordering on
$\mathbb{N}$, we can derive a more general ordering by removing these
tie-breakers and the need for a precedence ordering at all.

\begin{definition}[Knuth-Bendix Quasi-ordering (KBQO)]
  Given a weight function $w$, the Knuth-Bendix quasi-ordering $\geqslant_{kbo}$ is
defined as $t \geqslant_{kbo} u$ iff
$w(t) \geqslant w(u)$.
\end{definition}

The resulting quasi-ordering is considerably simpler to implement and is more
permissive: $t \kbo u$ implies $t \geqslant_{kbo} u$; and also enables arbitrary
associativity and commutativity axioms as rewrite rules, since it only considers
the weights of the function symbols and no structural components of the term.
However, one caveat is that \algoname operates on well-quasi-ordering that are
thin (\defref{orderings:prop}) and therefore can only consider KBQOs where
$w(f) > 0$ for all unary function symbols $f$.

However, the fact that KBO and KBQO largely ignore the structure of the term in
their comparison has a corresponding downside: it is not possible to
orient distributivity axioms, or  many other axioms that increase the
number of symbols in a term. Therefore, we have found that a WQO derived from
the recursive path ordering \cite{dershowitz1982orderings} to be more useful in
practice.

\subsubsection{Recursive Path Quasi-Orderings (RPQO)}
\label{sec:orderings:def}


In this section, we define a particular family of orderings designed to be
typically useful for term-rewriting via \algoname. Our family of orderings is a
novel extension of the classical notion of \rpo, designed to also be more
compatible with symmetrical rules such as commutativity and associativity (\cf{}
\problem{\challengeWQO}, \secref{overview}).

Like the classical \rpo notions,
our \emph{\rpqof{}} (\rpqo) is defined in three layers, derived from an
underlying ordering on function symbols:

\begin{itemize}
    \item The input ordering $\ord_\ops$ can be any quasi-ordering over $\ops$.
    \item The corresponding \emph{multiset quasi-ordering} $\ord_{M(X)}$ lifts an ordering $\ord_X$ over $X$ to an ordering $\ord_{M(X)}$
          over multisets of $X$.
          Intuitively $T \ord_{M(X)} U$ when $U$ can be obtained from $T$ by
          replacing zero or more elements in $T$ with the same number of equal
        (with respect to $\ord_{X}$) elements,
          and replacing zero or more elements in $T$ with a finite number of smaller ones
          (\defref{set:ordering}).
    \item Finally, the corresponding \emph{\rpqof} $\rpoc$ is an ordering over terms.
    Intuitively $f(ts) \rpoc g(us)$ uses $\ord_\ops$
    to compare the function symbols $f$ and $g$
    and the corresponding $\rpomc$ to compare the argument sets $ts$ and $us$
    (\defref{rpo}).
\end{itemize}

Below we provide the formal definitions of the multiset quasi-ordering and \rpqof
respectively generalized from
the multiset ordering of~\cite{dershowitz_proving_1979}
and the recursive path ordering~\cite{dershowitz1982orderings}
to operate on quasi-orderings.
For all the three orderings, we write
$x_l < x_r \doteq x_l \not \ord x_r$ and
$x_l > x_r \doteq x_l \ord x_r \land x_r \not \ord x_l$.

\begin{definition}[Multiset Ordering]
    \label{defn:set:ordering}
    Given a ordering $\ord_{X}$ over a set $X$, the \textit{derived multiset ordering} $\ord_{M(X)}$  over finite multisets of
$X$ is defined as $T \ord_{M(X)} U$ iff:
    \begin{enumerate}
            \item\label{defn:set:ordering:casei} $U = \emptyset$, or
            \item\label{defn:set:ordering:caseii}
            $t \in T \land u \in U \land t \approx u \land (T - t) \ord_{M(X)} (U - u)$, or
            \item\label{defn:set:ordering:casei}
            $t \in T \land (T - t) \ord_{M(X)} (U \setminus \{ u \in U~|~u <_{X} t \})$.
    \end{enumerate}
\end{definition}

\begin{definition}[Recursive Path Quasi-Ordering]\label{defn:rpo}\label{defn:rpqo}
Given a basic ordering $\ord_\ops$,
the \textit{\rpqof (\rpqo)} is the ordering $\rpoc$ over \trms{} defined as follows:
$f(\tt_1,\ldots, \tt_m) \rpoc  g(\uu_1, \ldots, \uu_n)$ iff
\begin{enumerate}
    \item\label{defn:rpo:casei} $f >_\ops g$  and $ \{ f(\tt_1,\ldots, \tt_m) \} \rpogtm {\{\uu_1, \ldots, \uu_n\}}$, or
    \item\label{defn:rpo:caseii} $g >_\ops f$  and $ \{\tt_1, \ldots, \tt_m\} \rpogeqm {\{  g(\uu_1, \ldots, \uu_n) \}}$, or
    \item\label{defn:rpo:caseiii} $f \approx g$ and $ \{\tt_1, \ldots, \tt_m\} \rpogeqm \{\uu_1,\ldots,\uu_n \}$.
\end{enumerate}
\end{definition}

\begin{example}
As a first example, any \rpqo $\ord_\trms$ used to restrict term rewriting will accept the rule
\goestor{}{X + Y}{Y + X}, since
$X + Y \ord_\trms Y + X$ always holds.
Since the top level function symbol is the same
$+ \approx +$, by \defcaseref{rpo}{caseiii}
we need to show $\{X,Y\} \rpogeqm \{Y,X\}$.
By \defcaseref{set:ordering}{caseii} (choosing both $t$ and $u$ to be $X$), we can reduce this to $\{Y\} \rpogeqm \{Y\}$; the same step applied to
$y$ reduces this to showing $\emptyset \rpogeqm \emptyset$ which follows directly from \defcaseref{set:ordering}{casei}.

\end{example}
From this example, we can see that both $X + Y \rpogeq Y + X$ and
$Y + X \rpogeq X + Y$ hold, in this case independently of the choice of input
ordering $\ord_\ops$ on function symbols. In our next example, the choice of
input ordering makes a difference.

\begin{example}
As a next example, we compare the terms
$s(X) + Y$ and $s(X + Y)$.
Now that the outer function symbols are \emph{not} equal,
the order relies on the ordering between $+$ and $s$.
Let's assume that $+ >_\ops s$.
Now to get $s(x) + y \rpogeq s(X + Y)$, the first case of Definition~\ref{defn:rpo}
further requires $\{ s(X) + Y \} \rpogtm \{ X + Y\}$, which holds if
$s(X) + y \rpogt X + Y$. The outermost symbol for both expressions is $+$, so we
must check the multiset ordering: $\stgt{\{s(X),Y\}}{\{X,Y\}}$, which holds
because by case splitting on the relation between $s$ and $X$, we can show that
$s(X)$ is always greater than $X$.
In short, if $+ >_\ops s$, then $s(X) + Y \rpogeq s(X + Y)$.
\end{example}

Developing on our $\rpqo$ notion (\defref{rpqo}), we consider the set of
\emph{all} such orderings that are generated by any total, well-quasi-ordering
over the operators. We prove that such term orderings satisfy the termination
requirements of \thmref{terminating}. Concretely:

\begin{theorem}
    If $\ord_\ops$ is a total, well-quasi-ordering,
    then
    \begin{enumerate}
        \item\label{thm:wqo} $\rpogeq$ is a well-quasi-ordering,
        \item\label{thm:thin} $\rpogeq$ is thin, and
        \item\label{thm:twf} $\rpogeq$ is thin well-founded.
    \end{enumerate}
\end{theorem}
\begin{proof}
    The detailed proofs can be found in \appref{ordering:proofs}.
    (\ref{thm:wqo}) uses
the well-foundedness theorem of
Dershowitz~\cite{dershowitz1982orderings}
and the fact that $\rpogeq$ is a
quasi-simplification ordering.
(\ref{thm:thin}) relies on the fact that
a finite number of function symbols can only generate
a finite number of equal terms.
(\ref{thm:twf}) is a corollary of
(\ref{thm:wqo}) and (\ref{thm:thin}) combined.
\end{proof}

\subsection{Ordering Constraint Algebras}
\label{sec:oca}

Ordering constraint algebras play a crucial role in the \algoname algorithm
(\secref{algo2}), by enabling the algorithm to simultaneously consider an entire
family of term orderings during the exploration of rewrite paths. In this
section, we provide a formal definition for ordering constraint algebras and
describe the construction of an algebra for the RPQO.

\begin{definition}[Ordering Constraint Algebra]\label{defn:oca}
An \textit{\Ocalgebra (OCA) $\ordconstraintstype_{(\tty,\grange)}$ over a set of
  terms $\tty$ and \zgedit{set of candidate term orderings}{term ordering family} $\grange$}, is a five-tuple $\ordconstraintstype_{(\tty,\grange)} \doteq \langle \octype, \gamma, \octop, \ocrefinec, \satc \rangle$, where:
\begin{enumerate}
\item $\octype$, the \textit{constraint language}, can be any non-empty set. Elements of $\octype$ are called \textit{constraints}, and are ranged over by $c$.
\item $\gamma$, the \textit{concretization function} of $\ordconstraintstype_{(\tty,\grange)}$, is a function from elements of $\octype$ to subsets of $\grange$.
\item $\octop$, the \textit{top constraint}, is a distinguished constant from $\octype$, satisfying $\gamma(\octop) = \grange$.
\item $\ocrefinec$, the \textit{refinement function}, is a function $\octype \rightarrow \tty \rightarrow \tty \rightarrow C$, satisfying (for all $c,\tt_l,\tt_r$) $\gamma(\ocrefine{c}{\tt_l}{\tt_r}) = \{ \ord \;|\; \ord~\in \gamma(c) \;\land\; {\tt_l}  \ord  {\tt_r}\}.$
\item $\satc$, the \textit{satisfiability function}, is a function $\octype \rightarrow \textit{Bool}$, satisfying (for all $c$) $\satc(c) = \textit{true} \;\Leftrightarrow\; \gamma(c)\neq\emptyset$.
\end{enumerate}
\end{definition}
The functions $\octop$, $\ocrefinec$, and $\satc$ are all called from our \acronym algorithm (\figref{algo-rewrite}), and must be implemented as (terminating) functions when implementing \acronym.
Specifically, \algoname instantiates the initial path with constraints
$c = \octop$. When a path can be extended via a rewrite application
$\goestor{\rrel}{t_{l}}{t_{r}}$, \algoname refines the prior path constraints $c$ to
$c' \doteq \ocrefine{c}{\tt_l}{\tt_r}$. Then, the new term is added to the path
only if the new constraints are satisfiable ($\satc(c')$ holds); that is, if $c'$ admits an
ordering that orients the generated path.
The function $\gamma$ need \emph{not} be implemented in practice; it is purely a mathematical concept used to give semantics to the algebra.

Given terms $\tty$ and a finite \zgedit{set of candidate orderings}{term
ordering family} $\grange$, a trivial OCA is obtained by letting $\octype = \mathcal{P}(\grange)$, and making $\gamma$ the identity function; straightforward corresponding elements $\octop$, $\ocrefinec$, and $\satc$ can be directly read off from the constraints in the definition above.

However, for efficiency reasons (or in order to support potentially infinite
sets of \zgout{candidate} orderings, which our theory allows), tracking these sets
symbolically via some suitably chosen constraint language can be preferable. For
example, consider lexicographic orderings on pairs of constants, represented by
a set $\tty$ of terms of the form $p(q_1,q_2)$ for a fixed function symbol $p$
and $q_1,q_2$ chosen from some finite set of constant symbols $Q$.
We choose the \zgedit{candidate orderings}{term ordering family} $\grange = \{ \ord_{\textit{lex}(\ord)}\;\mid\;\ord \textit{is a total order on }Q\}$ writing $\ord_{\textit{lex}(\ord)}$ to mean the corresponding lexicographic ordering on $p(q_1,q_2)$ terms generated from an ordering $\ord$ on $Q$.

A possible OCA over these $\tty$ and $\grange$ can be defined by choosing the constraint language $\octype$ to be \emph{formulas}: conjunctions and disjunctions of atomic constraints of the forms $q_1 > q_2$ and $q_1 = q_2$ prescribing conditions on the underlying orderings on $Q$. The concretization $\gamma$ is given by $\gamma(c) = \{ \ord_{\textit{lex}(\ord)}\;\mid\;\ord \textit{ satisfies }c\}$, \ie{} a constraint maps to all lexicographic orders generated from orderings of $Q$ that satisfy the constraints described by $c$, defined in the natural way. We define $\octop$ to be \eg{} $q = q$ for some $q\in Q$. A satisfiability function $\satc$ can be implemented by checking the satisfiability of $c$ as a formula\asfootnote{Do we need to say more about how?}. Finally, by inverting the standard definition of lexicographic ordering, we define:
 \[ \ocrefine{c}{p(q_{1},q_{2})}{p(r_{1},r_{2})} = c \land (q_1 > r_1 \lor (q_1 = r_1 \land q_2 > r_2) )\]

Using this example algebra, suppose that \acronym explores two potential rewrite steps $p(a_1,a_2) \rightarrow p(b_1,a_2) \rightarrow p(a_1,a_1)$. Starting from the initial constraint $c_0 = \octop$,
the constraint for the first step $c_1 \doteq \ocrefine{c_0}{p(a_1,a_2)}{p(b_1,a_2)}
=  a_1 > b_1 \lor (a_1 = b_1 \land a_2 > a_2)$
is satisfiable, \eg for any total order for which $a_1 > b_1$. However, considering the subsequent step, the refined constraint
$c_2 \doteq \ocrefine{c_1}{p(b_{1},a_2)}{p(a_{1},a_{1})}$, computed as
$c_2 = c_1\land (a_2 > a_2 \lor (a_2 = a_2 \land b_1 > a_1))$
 is no longer satisfiable. Note that this allows us to conclude that there is no
lexicographic ordering allowing this sequence of two steps, even without explicitly
constructing any orderings.

We now describe an OCA for \rpqos (\secref{orderings:def}), based on a compact
representation of sets of these orderings.


\subsubsection{An \Ocalgebra for $\rpogeq$}\label{sec:orderings:oc}


The OCA for RPQOs enables their usage in \algoname's
proof search. One simple but computationally intractable approach would be to
enumerate the entire \zgedit{family}{set} of \rpqos that orient a path; continuing the path so
long as the set is not empty. This has two drawbacks. First, the number of
\rpqos grows at an extremely fast rate with respect to the number of function
symbols; for example there are $6,942$ \rpqos describing five function symbols,
and $209,527$ over six \cite{numwqos}. Second, most of these orderings differ in ways that are
not relevant to the comparisons made by \algoname.

Instead, we define a language to succinctly describe the set of candidate
\rpqos, by calculating the minimal constraints that would ensure orientation of
the path of terms; \algoname continues so long as there is some \rpqo that
satisfies the constraints. Crucially the satisfiability check can be performed
effectively using an SMT solver, as described in \secref{impl-oca},
without actually instantiating any orderings.

Before formally describing the language, we begin with some examples, showing
how the ordering constraints could be constructed to guide the termination check
of \algoname.

\begin{example}[Satisfiability of Ordering Constraints]
Consider the following rewrite path given by the rules
$r_{1} \doteq f(g(X), Y) \rightarrow g(f(X, X))$ and $ r_{2} \doteq f(X, X) \rightarrow f(k, X)$:
$$f(g(h), k) \rightarrow_{r_{1}} g(f(h, h)) \rightarrow_{r_{2}} g(f(k, h))$$

To perform the first rewrite \algoname has to ensure that there exists an \rpqo
$\rpogeq$ such that $f(g(h), k) \rpogeq g(f(h, h))$.
Following \zgedit{the}{from} Definition~\ref{defn:rpo}, we obtain three possibilities:
\begin{enumerate}
    \item\label{cons:casei} $f >_\ops g$  and $ \{ f(g(h), k) \} \rpogtm {\{ f(h, h) \}}$, or
    \item\label{cons:caseii} $g >_\ops f$  and $ \{g(h), k\} \rpogeqm {\{  g(f(h, h)) \}}$, or
    \item\label{cons:caseiii} $f \approx g$ and $ \{g(h), k\} \rpogeqm \{f(h, h) \}$.
\end{enumerate}
We can further simplify these using the definition of the multiset quasi-ordering (\defref{set:ordering}).
Concretely, the multiset comparison of (\ref{cons:casei}) always holds, while
the multiset comparisons of (\ref{cons:caseii}) and (\ref{cons:caseiii}) reduce to $k >_\ops f \land k >_\ops g \land k >_\ops h$.
Thus, we can define the exact constraints $c_{0}$ on $\rpogeq$ to satisfy
$f(g(h), k) \rpogeq g(f(h, h))$ as
$$
c_{0} \doteq f >_\ops g  \lor (k >_\ops f \land k >_\ops g \land k >_\ops h)
$$
Since there exist many quasi-orderings satisfying this formula
(trivially, the one containing the single relation $f >_{\ops} g$), the first
rewrite is satisfiable.

Similarly, for the second rewrite, the comparison
$g(f(z, z)) \rpogeq g(f(k, z)) $ entails the constraints
$c_{1} \doteq z \ord_\ops k$. To perform this second rewrite the conjunction of
$c_{0}$ and $c_{1}$ must be satisfiable. Since the second disjunct of $c_0$
contradicts $c_1$, the resulting constraints $f >_\ops g \land z \ord_\ops k$
is satisfiable by an \rpqo, thus the path is satisfiable.
\end{example}

\begin{example}[Unsatisfiable Ordering Constraint]
As a second example, consider the rewrite rules $r_{1} \doteq f(x) \rightarrow g(s(x))$
and $r_{2} \doteq g(s(x)) \rightarrow f(h(x))$. These rewrite rules can clearly cause
divergence, as applying rule $r_{1}$ followed by $r_{2}$ will enable a
subsequent application of $r_{1}$ to a larger term. Now let's examine how our
\oca can show the unsatisfiability of the diverging path:
$$ f(z) \rightarrow_{r_{1}} g(s(z)) \not \rightarrow_{r_{2}} f(h(z)) $$
$f(z) \rpogeq g(s(z))$ requires $c_0  \doteq f > g \land f > s$ which is satisfiable, but
$g(s(z)) \rpogeq f(h(z))$ requires
 $c_1  \doteq (g \geqslant f \land g \geqslant h) \lor
  (g \geqslant f \land s \geqslant h) \lor
  (s > f \land s > h)$, which, although satisfiable on it's own, conflicts
  with $c_0$. Since no $\rpqo$ can satisfy both $c_0$ and $c_1$,
  the rewrite path is not satisfiable.
\end{example}

Having primed intuition through the examples, we now present a way to compute
such constraints. First, it is clear that we can define an \rpqo based on the
precedence over symbols $\ops$. Therefore, we define our language of constraints
to include the standard logical operators as well as atoms representing the
relations between elements of $\ops$, as:
$$
C_\ops \doteq f >_{\ops} g~|~f \approx g~|~C_{\ops} \land C_{\ops}~|~C_{\ops}
\lor C_{\ops}~|~\top~|~\bot
$$

Next, we lift our definition of $\rpqo$ and the multiset
quasi-ordering to derive functions:
$\rpqoc : \trms \rightarrow \trms \rightarrow C_{\ops}$,  and
$\mulc : (\trms \rightarrow \trms \rightarrow C_{\ops}) \rightarrow M(\trms) \rightarrow M(\trms) \rightarrow C_{\ops}$.
$\rpqoc$ is derived by a straightforward translation of \defref{rpo}:
\[
\begin{array}{rrcl}
\rpqoc{f(\tt_1,\ldots, \tt_m)}{g(\uu_1, \ldots, \uu_n)} = &
        f >_\ops g
        &\land&
        \mulcs{\rpqoc}{\{ f(\tt_1,\ldots, \tt_m) \}}{\{\uu_1, \ldots, \uu_n\}}\;\lor \\
&        g >_\ops f
        &\land& \mulc{\rpqoc}{\set{\tt_1, \ldots, \tt_m}}{\set{g(\uu_1, \ldots, \uu_n)}}
        \;\lor \\
&        f \approx g
        &\land&
        \mulc{\rpqoc}{\{\tt_1, \ldots, \tt_m\}}{\{\uu_1,\ldots,\uu_n \}}
\end{array}
\]
\noindent
where $\mulcs$ is the strict multiset comparison:
$\mulcs{f}{T}{U} = \mulc{f}{T}{U}~\land~\neg \mulc{f}{U}{T}$.
$\neg : C_{\ops} \rightarrow C_{\ops}$ inverts the constraints, with
$\neg(f >_{\ops} g) = f \approx g \lor g >_{\ops} f$ and
$\neg(f \approx g) = f >_{\ops} g \lor g >_{\ops} f$; the other cases are
defined in the typical way.

The definition for $\mulc$ is more complex. Recall that
$T \ord_{M(X)} U$ when $U$ can be obtained from $T$ by replacing zero or more
elements in $T$ with the same number of equal (with respect to $\ord_{X}$)
elements, and by replacing zero or more elements in $T$ with a finite number of
smaller ones. Therefore each justification for
$\srange{t}{m} \ord_{M(X)} \srange{u}{n}$ can be represented by a bipartite
graph with nodes labeled \range{t}{m} and \range{u}{n}, such that:
\begin{enumerate}
  \item Each node $u_{i}$ has exactly one incoming edge from some node $t_{j}$\as{.}
  \item If a node $t_{i}$ has exactly one outgoing edge, it is labeled either
\texttt{GT} or \texttt{EQ}.
  \item If a node $t_{i}$ has more than one outgoing edge, it is labeled
\texttt{GT}.
\end{enumerate}

$\mulc{f}{\srange{t}{m}}{\srange{u}{n}}$ generates all such graphs: for each
graph converts each labeled edge $(t, u, \texttt{EQ})$ to the formula
$f(t, u) \land f(u, t)$, each edge $(t, u, \texttt{GT})$ to the formula
$f(t, u) \land \neg f(u, t)$, and finally joins the formulas for the graph via a
conjunction. The resulting constraint is defined to be the disjunction of the
formulas generated from all such graphs.

Having defined the lifting of \as{recursive path quasi-orderings} to the language
of constraints, we define our \oca
$\rpqooc$ as the tuple
${\langle}C_{\ops},\top,\ocrefine,\gamma,\sat\rangle$ where:
\begin{itemize}
    \item $\ocrefine{c}{t}{u} = c \land \rpqoc{t}{u}$,
    \item $\grange$ is the set of all \rpqos,
    \item $\gamma(c)$ is the set of \rpqos derived from the underlying quasi-orders
        $\ord_{\ops}$ that satisfy $c$, and
    \item $\sat(c) = $ \textit{true} if \zg{and only if} there exists a quasi-order $\ord_{\ops}$
        satisfying $c$\zgout{, \textit{false} otherwise}.
\end{itemize}
\zgedit{We can show that $\rpqooc$ is an OCA, \ie satisfies the requirements of
-Definition~\ref{defn:oca}}{That $\rpqooc$ is an OCA, \ie satisfies the requirements of
\defref{oca}, follows by construction.
Namely, the function $\rpqoc{t, u}$ produces constraints $c$ such that,
for any RPQO $\rpogeq$, $t \rpogeq u$ if and only if its underlying ordering
$\ord_{\ops}$ satisfies $c$.}
In \secref{impl-oca} we further discuss how the satisfiability check is
mechanized and implemented using an SMT solver.


Having shown that using \rpqos as a term ordering is useful for theorem
proving, satisfies the necessary properties for \algoname, and admits an
efficient \oca, we continue our formal work by stating
and proving the metaproperties of \algoname.

\section{\algoname Metaproperties: Soundness, Completeness, and Termination}
\label{sec:meta}\label{sec:metaprop}
We now present the metaproperties of the
\algoname algorithm defined in~Figure~\ref{fig:algo-rewrite}.
We show
correctness (\thmref{correctness}),
completeness (\thmref{relative:completeness})
relative to the \as{input term ordering family (recall that its instances must all be thin well-quasi-orderings)}\asout{relations checked},
and termination (\thmref{terminating})
\as{which requires that calls to the OCA functions used in the algorithm, as well as the external oracle function, themselves terminate. The property that the orderings are thin well-founded guarantees in particular that any \emph{duplicate-free} path (such as those that \acronym{} generates) that can be oriented by any of these orderings is guaranteed to be finite.}
\asout{checked ordering relations be
decidable and well-founded on duplicate-free paths (no term occurs more than once, \eg{} those that \acronym{} generates).}
\as{We provide here the key invariants and statements of the formal results, and relegate the detailed proofs to} \appref{metaprop:proofs}.

\subsection{Formal Definitions}
\label{sec:background}
\label{defn:sn}
Our formalism of rewriting is standard; based on the terminology
of~\cite{KlopTRS}. Our language consists of the following:
\newcommand{\XX}{\ensuremath{X}\xspace}
\newcommand{\YY}{\ensuremath{Y}\xspace}
\newcommand{\xx}{\ensuremath{x}\xspace}
\newcommand{\yy}{\ensuremath{y}\xspace}
\begin{enumerate}[leftmargin=*]
  \item An infinite set of meta-variables (the variables for rewrite rules) \vars with elements $\XX$, $\YY$, \ldots.
  \item A finite set of \as{function symbols} \ops with elements $f$, $g$, \ldots
    Each operator is associated with a fixed numeric arity and types for its arguments and result (elided here, for simplicity).
  \item A set of terms \trms with elements $\tt,\uu,\ldots$ inductively defined as follows:
      (a) $\XX \in \vars \Rightarrow \XX \in \trms$ and
      (b) $f \in \ops$, $f$ has arity $n$, $\tt_1,\ldots,\tt_n \in \trms \Rightarrow f(\tt_1,\ldots,\tt_n) \in \trms$.
\end{enumerate}
We use \FV{\tt} to refer to the set of meta-variables in \tt. A term $\tt$ is \emph{ground} if $\FV{\tt}=\emptyset$.

  A \textit{substitution} $\sub \subseteq \vars \times \trms$ is a mapping from
  meta-variables to terms. We write \appsub{\sub}{\tt} to denote the
  simultaneous application of the substitution: namely, \appsub{\sub}{\tt} replaces
  each occurrence of each meta-variable $\XX$ in $\tt$ with $\sub(\XX)$.
%
A substitution $\sub$ \textit{grounds} $\tt$ if, for all $X\in\FV{\tt}$, $\sub(\XX)$ is a ground term.
  A substitution $\sub$ \textit{unifies} two terms \tt and \uu if $\appsub{\sub}{\tt} = \appsub{\sub}{\uu}$.

  A \textit{context} \ctx is a term-like object that contains exactly one
  term placeholder $\bullet$. If $t$ is a term, then $\inctx{t}$ is the term generated
  by replacing the $\bullet$ in $\ctx$ with $t$.

%
A \textit{rewrite rule} \rrule is a pair of terms $\rrule \doteq (\tt,\uu)$ such that
   $\FV{u} \subseteq \FV{t}$ and
   $t \notin \vars$.
%
Each rewrite rule $\rrule \doteq (\tt,\uu)$ defines a binary relation
$\rightarrow_{\rrule}$ which is the smallest relation such that, for all
contexts $\ctx$ and substitutions $\sub$ grounding $\tt$ (and therefore $\uu$),
$\goestor{\rrule}{\inctx{\appsub{\sub}{\tt}}}{\inctx{\appsub{\sub}{\uu}}}$.
%
%
%
%

We use \rrel to range over sets of rewrite rules.
We write \goestor{\rrel}{v}{w} iff \goestor{\rrule}{v}{w} for some $\rrule\in\rrel$.

For oracle functions (\as{from} terms to sets of terms) \evalSymb,
we write \goestoe{\tt}{\tt'} \textit{iff} $\tt'\in\evalSymb(\tt)$.
We write \goestor{\rrel+\evalSymb}{\tt}{\tt'} if
\goestor{\rrel}{\tt}{\tt'} or \goestor{\evalSymb}{\tt}{\tt'}.
For a relation $\rightarrow$ we write $\rightarrow^{*}$ for its reflexive, transitive closure.
%
  A \textit{path} is a list of terms. A binary relation \ord \textit{orients} a path $\tt_1,\ldots,\tt_n$ if $\forall
  i, 1 \leq i < n, \tt_i~\ord~\tt_{i+1}$.

\subsection{Soundness}
\label{subsec:meta:correct}

Soundness of \algoname means that any term of the output
($u \in \algod{\tt_0}$)
can be derived from the original
input term by \as{some} combination of term rewriting steps from $\rrel$ and steps via the oracle function $\evalSymb$ (\as{in other words,} $\evalstor{\rrel + \evalSymb}{\tt_0}{u}$).

Our proof relies on the following simple invariant of \algoname: any path stored in the stack during the execution of the algorithm
  can be derived by the rewrite rules in \rrel or the external oracle \evalSymb.
  \begin{restatable}[Path
    Invariant]{invariant}{pathinvariant}\label{invariant:rewrite_P}
    For any execution of \algod{\tt_0}, at the start of any iteration of the main loop, for each $(ts, c) \in p$, the list $ts$
      is a path of $\rrel + \evalSymb$ starting from $t_0$.
  \end{restatable}
  \begin{proof}
    (Sketch:)
    By straightforward induction on iterations of the main loop.
  \end{proof}

\begin{theorem}[Soundness of \algoname]\label{thm:correctness}
  For all \rrel, $\uu$, and $\tt_0$, if $u \in \algod{\tt_0}$, then $\evalstor{\rrel + \evalSymb}{\tt_0}{u}$.
\end{theorem}
\begin{proof}
  In each iteration of $\algoname$, the term $t$ added to the output $o$ is the
  last element of the list $ts$ for the tuple $(ts, c) \in p$. By Invariant~\ref{invariant:rewrite_P}, $t$
  must be on the path of $\rrel + \evalSymb$ starting from $\tt_0$.
\end{proof}

\subsection{Completeness}
\label{subsec:meta:completeness}
A \naive{} completeness statement for \algoname might be that, for any terms $\tt_0$ and $u$, if $\evalstor{\rrel + \evalSymb}{\tt_0}{u}$ then $u$ is in our output ($u \in \algod{\tt_0}$). This result doesn't hold in general by design, since \algoname explores only paths permitted by at least one \as{candidate instance of its input term ordering family.}\asout{of its input candidate orderings.} We prove this \emph{relative} completeness result in two stages.
First (\thmref{rewrite_complete_eval}),
we show that completeness always holds if all steps
only involve the external oracle.
Then (\thmref{relative:completeness}),
we prove relative completeness of \algoname with respect to the \as{provided term ordering family.}\asout{relation.}
We begin by stating another simple invariant of our algorithm: that any term appearing
  in a path in the stack $p$, will belong to the final output:
  \begin{invariant}\label{invariant:output}
    For any execution of \algod{\tt_0}, at the start of any iteration of the main loop,
    if $\tt \in ts$ and $(ts,c) \in p$, then, when the algorithm terminates, we will have $\tt \in \algod{\tt_{0}}$.
  \end{invariant}
  \begin{proof}
  (Sketch:) We can prove inductively that terms contained in any list in $p$ either remain in $p$ or end up in $o$; since $p$ is empty on termination, the result follows.
  \end{proof}

\begin{restatable}[Completeness w.r.t.~\evalSymb]{theorem}{relcompleteness}
  \label{thm:rewrite_complete_eval}
  For all \rrel, $\uu$, and $\tt_0$,
  if $\evalstor{\evalSymb}{\tt_0}{u}$, then $u \in \algod{\tt_{0}}$.
\end{restatable}
\begin{proof}
  (Sketch:) Since, $\evalSymb$ is strongly normalizing, the path of terms
  $\tt_0 \rightarrow_{\evalSymb} \ldots \rightarrow_{\evalSymb} u$ will
  not contain any duplicates; \algoname will therefore insert each term in the
  path into $ts$. Since $u$ is in that path, Invariant \ref{invariant:output}
  ensures $u \in \algod{\tt_{0}}$.
\end{proof}

\begin{restatable}[Relative Completeness]{theorem}{relcompletenesstwo}
  \label{thm:relative:completeness}
  For all \rrel, $\uu$, and $\tt_0$,
  if \evalstor{\rrel + \evalSymb}{\tt_0}{\uu}
  and there exists an ordering $\ord~\in\gamma(\octop)$ that orients the path justifying \evalstor{\rrel + \evalSymb}{\tt_0}{\uu}, then
$\uu \in \algod{t_{0}}$.
\end{restatable}

\begin{proof}
  (Sketch:)
  The proof structure is similar to Theorem~\ref{thm:rewrite_complete_eval}; in
  this case the terms in the path are guaranteed to be in $ts$ because some
  ordering in $\gamma(\octop)$ can orient the path.
\end{proof}

\subsection{Termination}\label{subsec:termination}
Termination of \algoname requires appropriate conditions on \zgout{the candidate
orderings,} the external oracle $\evalSymb$ and the ordering constraint algebra
$\ordconstraintstype$ employed. We formally define these requirements and then
prove termination of \algoname.

\begin{definition}[Well-Founded \oca{}s]\label{defn:oc_wf}
  For ordering constraint algebras $\ordconstraintstype =
  \langle\octype,\octop,\ocrefine, \satc, \gamma\rangle$,
  for $c, c' \in C$, we say $c'$ \textit{strictly refines} $c$ (denoted
  $c' \sqsubset_{\ordconstraintstype} c$) if $c' = \ocrefine{c}{t}{u}$ for some terms $t$ and $u$,
  and $\gamma(c') \subset \gamma(c)$. Then, we say $\ordconstraintstype$ is
  \textit{well-founded} if $\sqsubset_{\ordconstraintstype}$ is.
\end{definition}
Down every path explored by \acronym, the tracked constraint is only ever refined; well-foundedness of $\ordconstraintstype$ guarantees that finitely many such refinements can be strict.

We note that if the OCA describes a finite set of orderings, then it is
trivially well-founded: $\subset$ is well-founded on finite sets. For example,
the ordering constraint algebra for RPQOs (\secref{orderings:oc}) is
well-founded when the set of functions symbols $\ops$ is finite, as there are a
only a finite number of possible RPQOs over a finite set of function symbols.

\begin{definition}
  A relation \goestor{}{\tt_l}{\tt_r} is \textit{normalizing} if
  it does not admit an infinite path and \textit{\bounded}
  if for each $\tt_l$ it only admits finite $\tt_r$.
  \zgout{A relation \ord is \textit{thin well-founded} if it cannot orient a duplicate-free infinite path.}
\end{definition}

\begin{restatable}[Termination of \algoname]{theorem}{thmterminating}
  \label{thm:terminating}
  For any finite set of rewriting rules \rrel, if:
  \begin{enumerate}
    \item\label{tass:eval} $\goestoesymb$ is normalizing and \bounded,
    \item\label{tass:funcs} The $\ocrefine$ and $\satc$ functions from $\ordconstraintstype$ are decidable (always-terminating, in an implementation),
    \item\label{tass:wf} $\ordconstraintstype$ is well-founded,
  \end{enumerate}
  then, for all terms $\tt_0$, $\algod{\tt_0}$ terminates.
\end{restatable}
\begin{proof}
  (Sketch:)
The paths constructed by \algoname implicitly constructs a finitely branching
tree, and the four restrictions ensures that all paths down the tree are finite.
This ensures that the resulting tree is finite; and thus that \algoname's
implicit construction of the tree will terminate.
\end{proof}

Note that any deterministic, terminating external oracle function satisfies the
first requirement. Having completed the formalization, we now move on to the
details of our implementation.

\section{Implementation of \acronym}
\label{sec:implementation}

We implemented \algoname as a standalone library, comprising 2337 lines of
Haskell code (\secref{implementation:lib}). Our implementation includes the \algoname
algorithm, several \oca implementations (including RPQOs [\defref{rpo}])
and exposes the API for implementing \oca{}s (\secref{impl-oca}). We integrated this library into
the Liquid Haskell program verifier \cite{LiquidHaskell} (\secref{impl-lh}), where we chose the
task of applying \emph{lemmas} in Liquid Haskell proofs as a suitable target
problem for automation via \algoname.

\subsection{The \acronym Library}
\label{sec:implementation:lib}

Our \acronym~implementation is developed in Haskell and can be used directly by
other Haskell projects. The library is designed modularly; for example, a client
of the library can decide to use \algoname only for comparing terms via an OCA,
without also using the proof search algorithm of \secref{algo2}. In addition,
our library has a small code footprint and \as{can be used with or without external solvers}\asout{does not require
external solvers}, making it ideal for integration into existing program analysis
tools and theorem provers.

\as{Furthermore, we} include in the library built-in helper utilities for encoding and
solving constraints on term orderings. Although the library enables integration
of arbitrary solvers; it provides several built-in solvers for constraints on
finite WQOs and also provides an interface for solving constraints with external
SMT solvers. These utilities comprise the majority of the code in the
\acronym~library (1369 out of the 2337 lines).

Our implementation defines the OCA interface of \secref{oca} and provides
three built-in instances for RPQOs, LPQOs (derived from the Lexicographic path
ordering), and KBQOs (\secref{kbqo}). The helper utilities included in the
library enable a concise implementation of these OCAs: the three OCA implementations consist of 200 lines of code in total.

To facilitate debugging and evaluation of OCAs, the library also provides a
standalone executable that produces visualizations of the rewrite paths that
\algoname~explores when using the OCA to compute the rewrites paths from a given
term. \figref{vis-example1} and \figref{impl:opt} were produced using this
functionality; we also note that the visualization is also capable of displaying
the accumulated constraints on the ordering at each node in the tree.

We now describe the interface for defining OCAs in our \acronym implementation,
via a presentation of the RPQO algebra in the library.

\subsection{Efficient Implementations of OCAs in \acronym}\label{sec:impl-oca}

\begin{figure}[t]
\begin{minipage}{0.485\textwidth}
\begin{mcode}
  -- Interface of OC Algebra
  data OC C T = OC
    { top    :: C
    , refine :: C -> T -> T -> C
    , sat    :: C -> IO Bool
    }

  -- Language of Logical Formulas
  data LF A = LTrue  | LFalse
          | A :>: A  | A :=: A
          | LF A :$\land$: LF A | LF A :$\lor$: LF A
\end{mcode}
\end{minipage}%
\begin{minipage}{0.515\textwidth}
\begin{mcode}
  -- Implementation of OC Algebra
  rpoOC :: OC (LF $\ops$) $\trms$
  rpoOC = OC LTrue refine sat where

    refine :: LF $\ops$ -> $\trms$ -> $\trms$ -> LF $\ops$
    refine c t u =
      c :$\land$: rpo t u -- As in Def 10

    sat :: LF $\ops$ -> IO Bool
    sat = smtSat . toSMT -- SMT Interface
\end{mcode}
\end{minipage}%

\caption{The implementation of our \rpqo Ordering Constraint Algebra}
\label{fig:impl:ocsat}
\end{figure}
Figure~\ref{fig:impl:ocsat} presents \algoname's library interface for ordering
constraint algebras and the implementation of RPQOs. The interface !OC! is parametric in the language of constraints !C! and the type of
terms !t!. The logical formulas !LF A! describe constraints on WQOs over !A!, in
the case of RPQOs, !LF $\ops$! tracks constraints on the underlying precedence of
function symbols.

Our implementation !rpoOC! defines the initial constraints !top! to
be !LTrue!, (intuitively, permitting any \rpqo). The function !refine c t u!
conjoins the current constraints !c! with the constraints !rpo t u!, ensuring
$t \ord u$. Finally the !sat! function converts the constraints into an
equisatisfiable SMT formula, by encoding each distinct function symbol as an SMT
integer variable, encoding the logical operators as their SMT equivalent, and
checking for satisfiability of the resulting formula.

\algoname's interface supports arbitrary implementations for ordering constraints and
is not dependent on any particular ordering, constraint language, or solver. For
example, the !sat! function for RPQOs could evaluate the formulas using an
alternative solver; in fact \algoname includes a built-in solver for this
purpose, although it does not achieve as high performance as the SMT-based approach.




\subsection{Integration of REST in Liquid Haskell}
\label{sec:impl-lh}
We used \algoname to automate \as{lemma application}\asout{theorem proving} in Liquid Haskell.
Here we provide \as{a brief}\asout{an} overview of Liquid Haskell (\secref{lh-lh}),
how \algoname is used to automate lemma instantiations (\secref{lh-rest})
and how it mutually interacts with the existing Liquid Haskell automation (\secref{lh-ple}).

\subsubsection{Liquid Haskell and Program Lemmas}\label{sec:lh-lh}
\begin{figure}[t]
\begin{mcode}
{-@ example1 :: s0 : Set -> { s1 : Set | IsDisjoint s0 s1 } -> f : (Set -> a) -> { f ((s0 \/ s1) /\ s0)  = f s0 } @-}
example1 :: Set -> Set -> (Set -> a) -> Unit
example1 s0 s1 f =
       f ((s0 \/ s1) /\ s0)         ? distribUnion s0 s1 s0
   === f ((s0 /\ s0) \/ (s1 /\ s0)) ? idemInter s0
   === f (s0 \/ (s1 /\ s0))         ? symmInter s1 s0
   === f (s0 \/ (s0 /\ s1))         -- Disjoint
   === f (s0 \/ emptySet)           ? emptyUnion s0
   === f s0
   *** QED
\end{mcode}
\caption{Liquid Haskell version of the proof from \exref{identities}.}\label{fig:exampleoneLH}
\end{figure}

Liquid Haskell performs program verification via \emph{refinement types} for
Haskell; function types can be annotated with refinements that capture
logical/value constraints about the function's parameters, return value and
their relation. For example, the function \texttt{example1} in~\figref{exampleoneLH}\asfootnote{The code in the
figure looks strange as the environment picks up /\ as a math op. Maybe we can
find a different listing environment?}
ports the
set example of \exref{identities}
to Liquid Haskell, without any use of \algoname.
User-defined lemmas amount to nothing more than additional program
functions, whose refinement types express the logical requirements of the lemma.
The first line of the figure is special comment syntax used in Liquid Haskell to
introduce refinement types; it expresses that the first parameter !s0! is
unconstrained, while the second !s1! is refined\asfootnote{I don't know whether
we need to explicitly introduce the syntax for refinements in general and
explain it.} in terms of !s0!: it must be some value such that
\texttt{IsDisjoint s0 s1} holds. The refinement type on the (unit) return value
expresses the proof goal; the body of the function provides the proof of this
lemma. The proof is written in equational style; the \texttt{?} annotations
specify lemmas used to justify proof steps \cite{vazou_theorem_2018}. The
penultimate step requires no lemma; the verifier can discharge it based on the
refinement on the \texttt{s1} parameter.

Lemmas already proven can be used in the proof of further lemmas; as is standard
for program verification, care needs to be taken to avoid circular reasoning.
Liquid Haskell ensures this via well-founded recursion: lemmas can only be
instantiated recursively with smaller arguments.

%
%

\subsubsection{\algoname for Automatic Lemma Application in Liquid Haskell}
\label{sec:impl:rewrite-rules}\label{sec:lh-rest}
We apply \algoname to automate the application of equality lemmas in the context of
Liquid Haskell. The basic idea is to extract a set of rewrite rules from a set
of refinement-typed functions, each of which must have a refinement type
signature of the following shape:
\smallskip
\begin{mcode}
{-@ rrule :: x$_1$:t$_1$ -> $\dots$ -> x$_n$:t$_n$ -> {v:() | $e_{l} = e_{r}$ } @-}
\end{mcode}
In particular, the equality !$e_{l} = e_{r}$! refinement of the (unit) return value generates potential rewrite rules to feed to \algoname, in both directions.
%
Let \FV{e} be the free variables of $e$,
if $\FV{e_r} \subseteq \FV{e_l}$ and $e_l \not \in \{x_1,\ldots,x_n\}$ then $e_l
\rightarrow e_r$ is generated as a rewrite rule. Symmetrically,
if $\FV{e_l} \subseteq \FV{e_r}$ and $e_r
\not \in \{x_1,\ldots,x_n\}$ then $e_r \rightarrow e_l$ is generated as a rewrite rule. These rewrite rules are fed to \algoname along with the current terms we are trying to equate in the proof goal; any rewrites performed by \algoname are fed back to the context of the verifier as assumed equalities.

Since the extracted rewrite rules are defined as refinement-typed expressions, our implementation technically goes beyond simple term rewriting, since instantiations of these rules in our implementation are also refinement-type-checked; \ie it instantiates only the rules with expressions of the proper refined type, achieving a
form of conditional rewriting \cite{kaplan1984}.

\subparagraph*{Selective Activation of Lemmas: Local and Global Rewrite Rules}
In our Liquid Haskell extension, the user can activate a rewrite rule globally or locally, using
the !rewrite! and !rewriteWith! pragmas, \resp. For example, with the below annotations
\smallskip
\begin{mcode}
  {-@ rewrite global @-}
  {-@ rewriteWith theorem [local] @-}
\end{mcode}
the rule !global! will be active when verifying every function in the current Haskell module,
while the rule !local! is used only when verifying !theorem!.

\subparagraph*{Preventing Circular Reasoning}
Our implementation finally ensures that rewrites cannot be used to justify
circular reasoning, by checking that there are no cycles induced by our !rewrite! and !rewriteWith! pragmas. For example, the below, unsound, circular dependency will be rejected
with a rewrite error by our implementation.
\smallskip
\begin{mcode}
  {-@ rewriteWith p1 [p2] @-}
  {-@ rewriteWith p2 [p1] @-}
  {-@ p1, p2 :: x:Int -> { x = x + 1 } @-}
  p1 _ = () ; p2 = p1
\end{mcode}
To prevent circular dependencies, we check that the dependency graph of the
rewrite rules (which are made available for proving with) has no cycles. This
simple restriction is stronger than strictly necessary; a more-complex
termination check could allow rewrites to be mutually justified by ensuring that
recursive rewrites are applied with smaller arguments. In practice, our coarse
check isn't too restrictive: because Haskell's module system enforces acyclicity
of imports, rewrite rules placed in their own module can be freely referenced
by importing the library.

\subparagraph*{Lemma Automation}
Using our implementation, the same \exref{identities} proven manually in \figref{exampleoneLH} can be alternatively proven (with all relevant rewrite rules in scope) as follows:
\smallskip
\begin{mcode}
{-@ example1 :: s0 : Set -> { s1 : Set | IsDisjoint s0 s1 } ->  f : (Set -> a) ->
                { f ((s0 \/ s1) /\ s0) = f s0 } @-}
example1 s0 s1 _ = ()
\end{mcode}
The proof is fully automatic: no manual lemma calls are needed as these are all handled by \algoname.
Integrating \algoname into Liquid Haskell
required around 500 lines of code, mainly for surface syntax.

%
%
%

\subsubsection{Mutual PLE and \algoname Interaction}
\label{subec:impl:ple}\label{sec:lh-ple}
\begin{figure}[t]
   \begin{center}
    \includegraphics[width=0.4\textwidth]{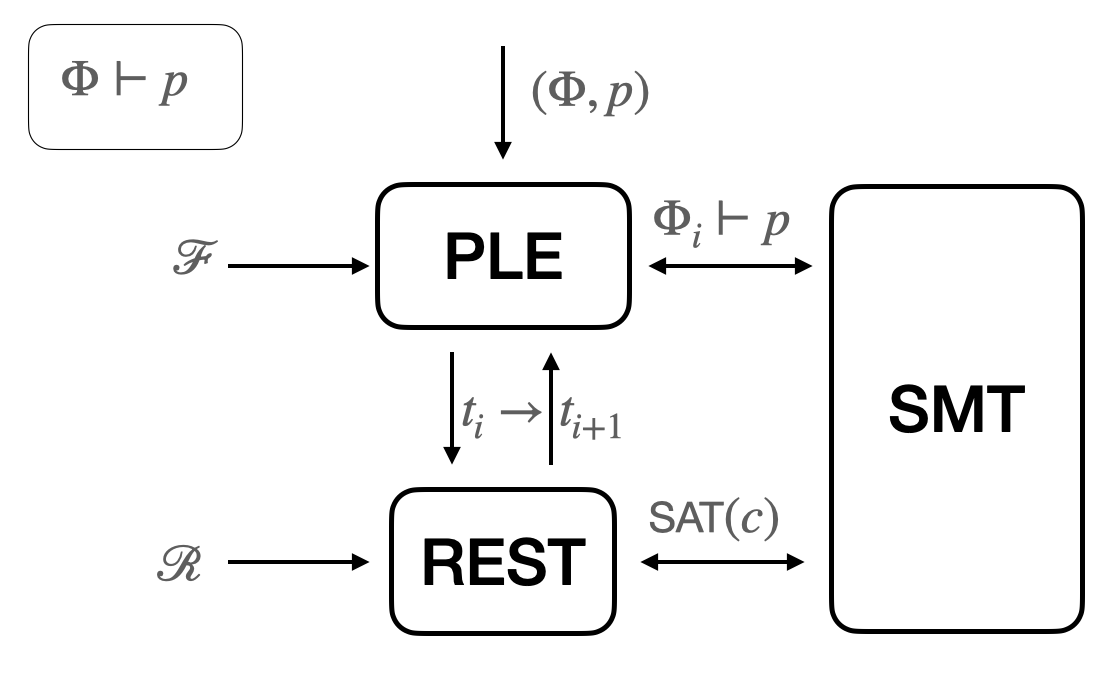}
  \end{center}
  \caption{Interaction between PLE and \algoname. }
  \label{fig:implementation}
\end{figure}
Liquid Haskell includes a technique called \emph{Proof by Logical Evaluation} (PLE) \cite{PLE} for automating the expansion of terminating program function definitions. PLE expands function calls into single cases of their (possibly conditional) bodies exactly when the verifier can prove that a unique case definitely applies. This check is performed via SMT and so can condition on arbitrary logical information; in our implementation, this forms a natural complement to the term rewriting of \algoname, and plays the role of its external oracle (\cf{} \secref{algo}). Since PLE is proven terminating \cite{PLE}, the termination of this collaboration is also guaranteed (\cf{} \secref{meta}).

Figure~\ref{fig:implementation} summarizes the mutual interaction between PLE and \algoname on a verification condition 
$\Phi \vdash p$, where $\Phi$
~is an environment of assumptions.
PLE also takes as input a set $\mathcal{F}$ of (provably) terminating, user-defined function definitions that it iteratively evaluates.
~Meanwhile, \algoname is provided with the rewrite rules extracted from in-scope lemmas in the program (\cf \secref{impl:rewrite-rules}); these two techniques can then generate paths of equal terms including steps justified by each technique.  For example, consider the following simple lemma \texttt{countPosExtra}, stating that the number of strictly positive values in \texttt{xs ++ [y]} is the number in \texttt{xs}, provided that \texttt{y <= 0}, and a lemma stating that \textit{countPos} of two lists appended gives the same result if their orders are swapped.
\smallskip
\begin{mcode}
{-@ lm :: xs : [Int] -> ys : [Int] -> { countPos (xs ++ ys) = countPos (ys ++ xs) } @-}

{-@ rewriteWith countPosExtra [lm] @-}
{-@ countPosExtra :: xs : [Int] -> {y : Int | y <= 0 } ->
                    { countPos (xs ++ [y]) = countPos xs } @-}
countPosExtra :: [Int] -> Int -> ()
countPosExtra _ _ = () -- proof is fully automatic!
\end{mcode}
The proof requires rewriting \texttt{countPos(xs ++ [y])} first via lemma \texttt{lm} (by \algoname), expanding the definition of \texttt{++} twice (via PLE) to give \texttt{countPos(y:xs)}, and finally one more PLE step evaluating \texttt{countPos}, using the logical fact that \texttt{y} is not positive. Note in particular that the first step requires applying an external lemma (out of scope for PLE) and the last requires SMT reasoning not expressible by term rewriting. The two techniques together allow for a fully automatic proof.

\subsection{Further Optimizing the \algoname Algorithm}
\label{subsec:impl:optim}

\begin{figure}[t]
  \begin{center}
  \includegraphics[width=0.6\textwidth]{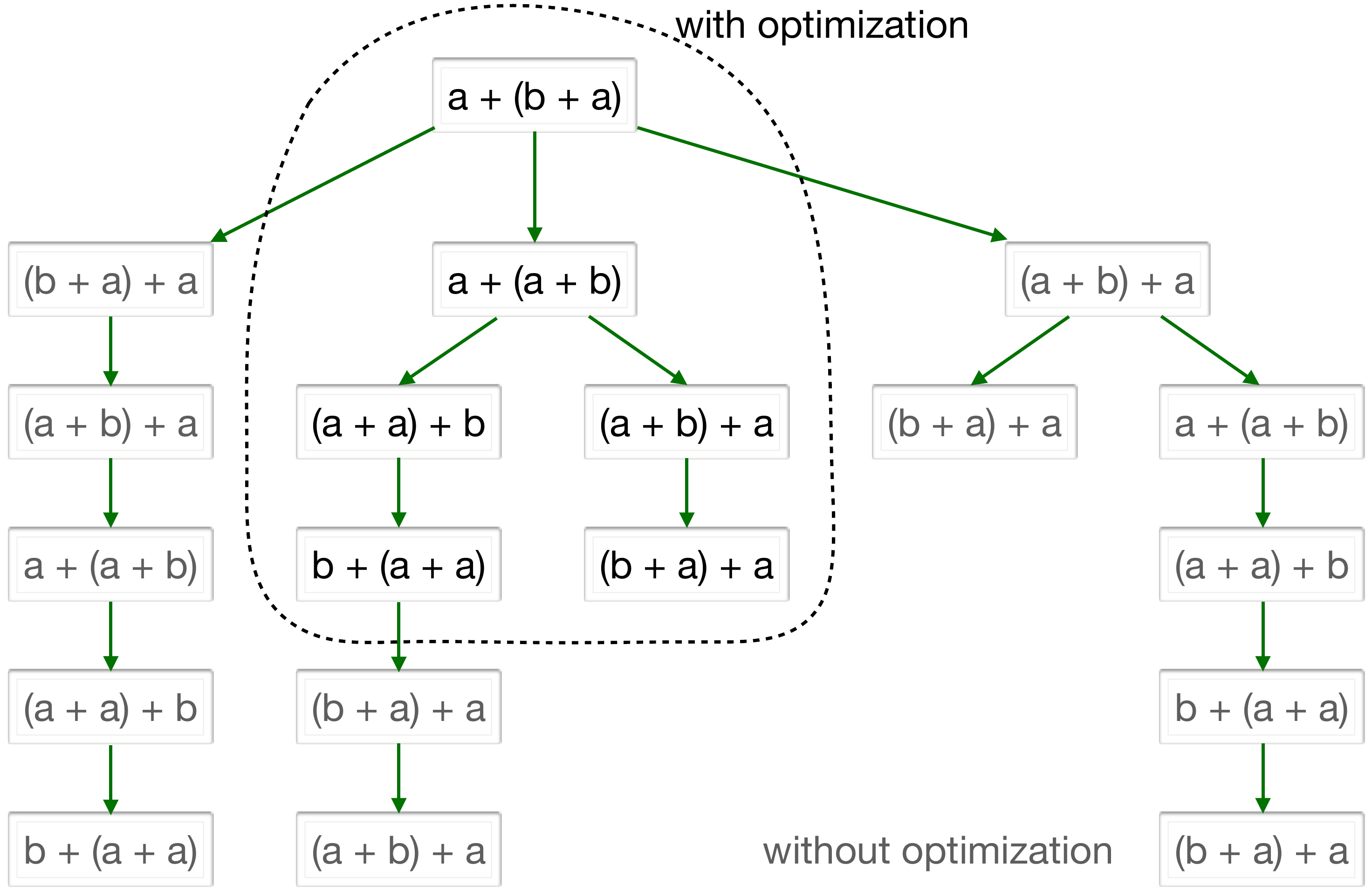}
  \end{center}
\caption{Associative-commutative rewrites of $a + (b + a)$ generated by
\algoname. Paths explored by \algoname with the explored terms optimization are
within the dashed line. Using the explored terms optimization, \algoname only
considers each term once.}
\label{fig:impl:opt}
\end{figure}
When a rewrite system is branching, \algoname may encounter different rewrite
paths from an initial term $t$ to an arbitrary term $u$. For example, in
Figure~\ref{fig:impl:opt}, the term $(b+a) + a$ is explored in 5 different
paths. In general, \algoname cannot always ignore the repeat encounters of $u$,
as \asout{the}\as{a new} path from $t$ to $u$ may impose ordering constraints enabling more
rewrites in the future. Nonetheless, reducing the number of explored paths \as{naturally improves}\asout{would
improve} performance. Therefore, we optimize \algoname based on the following
observations:

\begin{enumerate}
  \item A term $t$ does not need to be revisited if all of it's \asout{transitive
}rewrites have already been visited.
  \item If a term $t$ was previously visited at constraints $c$, revisiting $t$
        at constraints $c'$ is not necessary if $c$ permits all
        orderings permitted by $c'$, \ie $\gamma(c') \subseteq \gamma(c)$.
\end{enumerate}

To implement this optimization,
\algoname maintains a mapping $\cmap$ from terms to the logical constraints $c$ each term was explored with \as{(initially mapping all terms to }!top!).
To explore a term $\tt$  under logical constraints $c$, the algorithm checks that
this term is \emph{explorable}, formally defined \as{by:}\asout{as below.}
$$
\explorable{\tt}{c} \doteq \tt \not \in M \lor (\lnot (c \Rightarrow M[t])
\land \exists \uu. (\goestor{\rrel}{\tt}{\uu} \land \explorable{\uu}{c}))
$$
This predicate ensures that either this term was not explored before or
it comes with weaker constraints that can derive at least one new term in the path.

After exploring a new term, \algoname\ \as{weakens}\asout{updates} the mapping
$\cmap$ \as{for this term to the disjunction of the constraints under which it was newly explored and those previously mapped to in $\cmap$.}
%
With this optimization, a term will appear in more than one path in the \algoname graph
only when it can lead to different terms in the path\asfootnote{Not sure this is clear}.
This optimization critically reduces the number of explored terms \as{even for small examples:} as shown in Figure~\ref{fig:impl:opt}
where 19 vertices of the \algoname\ \as{graph}\asout{path} \asout{on the left}\as{shown} reduced to only \as{the} 6 \as{in the dotted region}\asout{on the right}.

\section{Evaluation}
\label{sec:evaluation}
Our evaluation seeks to answer three research questions:

\noindent
\textbf{\S~\ref{sec:eval:comparison}: How does \algoname compare to existing rewriting tactics?}

\noindent
\textbf{\S~\ref{sec:eval:ematching}: How does \algoname compare to E-matching based axiomatization?}

\noindent
\textbf{\S~\ref{sec:eval:simplification}: Does \algoname~simplify equational proofs?}

We evaluate \algoname using the Liquid Haskell implementation described in
\secref{implementation}. In \secref{eval:comparison}, we
compare our implementation's rewriting functionality with that of other theorem
provers, with respect to the challenges\nvtodo{make sure problems is changed to challenged} mentioned in \secref{challenges}. In
\secref{eval:ematching}, we compare against Dafny~\cite{dafny} by porting Dafny's
calculational proofs to Liquid Haskell, using rewriting to handle axiom
instantiation. Finally, in \secref{eval:simplification}, we port proofs from
various sources into Liquid Haskell both with and without rewriting, and compare
the performance and complexity of the resulting proofs.

\subsection{Comparison with Other Theorem Provers}
\label{sec:eval:comparison}
\begin{table}
  \begin{center}
    \begin{tabular}{|l| c|c|c|c|c|c|c |}
    \hline
    Property                    & \toolname  & \Coq      & \Agda & \Lean & \texttt{Isabelle} & \Zeno & \texttt{Isa+} \\
    \hline
    \texttt{Diverge}            & \checkmark 0.62s & loop      & loop        & fail      & loop       & \checkmark 0.47s & \checkmark 7.58s \\
    \texttt{Plus AC}            & \checkmark 1.13s & loop      & loop        & fail      & fail       & \checkmark 0.54s & \checkmark 4.30s \\
    \texttt{Congruence}         & \checkmark 0.69s & \checkmark 0.22s  & \checkmark 26.10s   & \checkmark 0.36s  & \checkmark 3.86s  & fail      & \checkmark 4.39s \\
    \hline
  \end{tabular}
 \end{center}
    \caption{\small{Comparison of \algoname with existing theorem provers.
    \toolname is Liquid Haskell with rewriting.
    The potential outcomes are
    \checkmark when the property is proved;
    \textbf{loop} when no answer is returned after 300
    sec; and
    \textbf{fail} when the property cannot be proven.
    \texttt{Isa+} is \Isabelle with \texttt{Sledgehammer}.}}
      \label{fig:tp_compare}
\end{table}

To compare \algoname with the rewriting functionality of other theorem provers,
we developed three examples to test the \numchallenges challenges described in
\secref{challenges} and compare our implementation to that of
other solvers. We chose to evaluate against \Agda~\cite{Agda}, \Coq\cite{Coq},
\Lean~\cite{Lean}, \Isabelle~\cite{Isabelle}, and \Zeno\cite{Zeno}, as they are
widely known theorem provers that either support a rewrite tactic, or use
rewriting internally.
\Agda, \Lean, and \Isabelle allow
user-defined rewrites. In \Lean~and \Isabelle, the tactic for applying
rewrite rules multiple times is called !simp!; for simplification.
\Agda, \Coq, and \Isabelle's
implementation of rewriting can diverge for nonterminating rewrite
systems~\cite{Coq,AgdaUserManual,Isabelle}. On
the other hand, \Lean~enforces termination, at least to some degree, by ensuring
that associative and commutative operators can only be applied according to a
well-founded ordering \cite{Lean_tp}.
\Zeno~\cite{Zeno} does not allow for user-defined rewrite
rules, rather it generates rewrites internally based on user-provided axioms.
\texttt{Sledgehammer}~\cite{meng2008translating,paulsson2012three,paulson2007source}
is a powerful tactic supported by \Isabelle that (on top of the built-in
rewriting) dispatches proof obligations to various external provers and succeeds
when any of the external provers succeed; this tactic operates under a built-in
(customizable) timeout.

1. \texttt{Diverge} tests how the prover handles the challenges 1 and \challengeOracle:
restricting the rewrite system to ensure termination and
integrating external oracle steps. This example encodes a single (terminating)
rewrite rule $f(X) \rightarrow g(s(s(X)))$ and terminating, mutually recursive function
definitions for $f$ and $g$. However, the combination of the rules and function
expansions can cause divergence. This test also requires a simple proof that
follows directly from the function definitions.

2. \texttt{Plus AC} tests the challenges \challengeDiffInstantiations and \challengeWQO
by encoding a task that requires a permissive term ordering.
This example encodes !p!, !q!, and !r!, user-defined natural
numbers, and requires that expressions such as !(p + q) + r! can be rewritten
into different groupings such as !(r + q) + p!, via associativity and
commutativity rules.

3. \texttt{Congruence} is an additional test to ensure that the implementation
of the rewrite system is permissive enough to generate the expected result. This
test evaluates a basic expected property, that the expressions $f(g(t))$ and
$f(g'(t))$ can be proved equal if there exists a rewrite rule of the form
$g(X) \rightarrow g'(X)$.

We present our results in Table~\ref{fig:tp_compare}. As expected, \Coq, \Agda,
and \Isabelle diverge on the first example, as they do not ensure termination of
rewriting. \Lean does not diverge, but it also fails to prove the theorem.
Unsurprisingly, the commutativity axiom of !Plus AC! causes theorem provers that
don't ensure termination of rewriting to loop. Although \Lean ensures
termination, it does not generate the necessary rewrite application in every
case, because it orients associative-commutative rewriting applications
according to a fixed order. With the exception of \Zeno, all of the theorem
provers tested were able to prove the necessary theorem for the final example.
Our implementation succeeds on these three examples by implementing a permissive
termination check based on non-strict orderings.

For this selection of simple but illustrative examples, the only tools to
succeed on all cases are our implementation, and Isabelle's Sledgehammer. The
latter combines a great many techniques which go beyond term rewriting.
Nonetheless, we note that our novel approach provides a clear and general formal
basis for incorporation with a wide variety of verifiers and reasoning
techniques (due to its generic definition and formal requirements), and provides
strong formal guarantees for such combinations. In particular, \acronym{}
provides general termination and relative completeness guarantees, which
Sledgehammer (via its timeout mechanism) does not.

\subsection{Comparison with E-matching}
\label{sec:eval:ematching}

To evaluate \algoname against the E-matching based approach to axiom
instantiation, we compared with Dafny \cite{dafny}, a state-of-the\nv{-}art
program verifier. Dafny supports equational reasoning via calculational proofs
\cite{leino2013verified} and calculation with user-defined functions
\cite{amin2014computing}. We ported the calculational proofs of
\cite{leino2013verified} to Liquid Haskell, using rewriting to automatically
instantiate the necessary axioms.

\subsubsection{List Involution}\label{sec:eval:listinv}

Figure \ref{fig:eval:involution} shows \as{an example taken directly from Dafny \cite{leino2013verified}}, proving that the reverse operation on lists is an
involution, \ie $\forall \text{xs} . \textit{reverse(reverse(xs))} = \textit{xs}$.
In this example, both Liquid Haskell and Dafny operate on inductively defined
lists with user-defined functions \concat and \texttt{reverse}. The \as{original Dafny proof goes
through via the combination of a manual application of a lemma called} !ReverseAppendDistrib! \as{(stating that for all lists $xs$ and $ys$,} $reverse(xs \concat ys) = reverse(ys) \concat reverse(xs)$)
and induction on the size of the list.

Using \as{term rewriting as enabled by our \algoname library}, Liquid Haskell is able to simplify the proof, with PLE
expanding the function definitions for !reverse! and !append!, and \algoname~
\as{applying}\asout{generating} the \as{necessary} equality \texttt{reverse (reverse xs ++ [x]) = reverse [x] ++ reverse
(reverse xs)}.

In Dafny, a similar simplification of the calculational proof is not possible\as{; the proof fails if the manual equality steps are simply removed}.
\as{We experimented further and found that} the lemma !ReverseAppendDistrib! can be \as{alternatively} encoded as a \as{user-defined} axiom \as{which, by itself, does not appear to cause trouble for E-matching, and with this change alone the proof succeeds without the need for this single lemma call. On the other hand}, the
equalities must \as{still} be mentioned \as{for} the calculational proof \as{to succeed}. \as{Perhaps surprisingly,} removing \as{these}
intermediate equality steps caused Dafny to stall\footnote{We include this
 version in \appref{evaluation:dafnyloop}}; analysis with the
Axiom Profiler \cite{BeckerMuellerSummers19} indicated the presence of a \as{(rather complex)}
matching loop involving the axiom !ReverseAppendDistrib! \as{in combination with axioms internally generated by the verifier itself}. \as{This illustrates that achieving further automation of such E-matching-based proofs is not straightforward, and can easily lead to performance difficulties due to matching loops which can be hard to predict and understand, even in this state-of-the-art verifier. By contrast, \algoname can automatically provide the necessary equality steps for this proof without introducing any risk of non-termination.}

\begin{figure}[t]
\begin{subfigure}[b]{\textwidth}
  \begin{mcode}
lemma LemmaReverseTwice(xs: List)
    ensures reverse(reverse(xs)) == xs;
{
  match xs {
    case Nil =>
    case Cons(x, xrest) =>
      calc {
        reverse(reverse(xs));
        reverse(append(reverse(xrest), Cons(x, Nil)));
        { ReverseAppendDistrib(reverse(xrest), Cons(x, Nil)); }
        append(reverse(Cons(x, Nil)), reverse(reverse(xrest)));
        { LemmaReverseTwice(xrest); }
        append(reverse(Cons(x, Nil)), xrest);
        append(Cons(x, Nil), xrest);
        xs;
      }
  }
}\end{mcode}
  \caption{Calculation-style proof in Dafny, from \cite{leino2013verified}.}
\end{subfigure}

\quad\\
\begin{subfigure}[b]{\textwidth}
  \begin{mcode}
{-@ involutionP :: xs:[a] -> {reverse (reverse xs) == xs } @-}
{-@ rewriteWith involutionP [distributivityP] @-}
involutionP []     = ()
involutionP (x:xs) = involutionP xs\end{mcode}
  \caption{An equivalent proof implemented in Liquid Haskell extended with \algoname}
\end{subfigure}
\caption{List Involution proofs in Liquid Haskell and Dafny}\label{fig:eval:involution}
\end{figure}

\subsubsection{Set Properties}\label{sec:eval:set}

Figure \ref{fig:eval:dafny:set} shows the Dafny and Liquid Haskell proofs for
the implication
$s_{0} \cap s_{1} = \emptyset \implies f((s_{0} \cup s_{1}) \cap s_{0}) = f(s_{0})$.

Dafny uses a calculational proof to show the equality
$(s_{0} \cup s_{1}) \cap s_{0} = s_{0}$, seemingly by applying distributivity\nvtodo{I do not see distributivity anywhere in the proof}.
In fact, the distributivity aspect is not relevant to the proof; rather, the set
equality in the proof syntax causes Dafny to instantiate the set extensionality
axiom discharging the proof. It is for this reason that Dafny requires an extra
proof step to prove $f((s_{0} \cup s_{1}) \cap s_{0}) = f(s_{0})$, as this term
does not include an equality on sets, but rather on applications of $f$. Dafny's
set axiomatization does not include the distributivity axiom, as \nvout{the presence of
}such an axiom could easily lead to matching loops.

Using \algoname, it is safe to encode arbitrary lemmas as rewrite rules, as the
termination is guaranteed; in this case the distributivity lemma can be used to
complete the proof (and is permitted as a rewrite rule with the precedence
$\cap > \cup$).

\begin{figure}[t]
\begin{subfigure}[b]{\textwidth}
  \begin{mcode}
lemma Proof<a>(s0: set<int>, s1: set<int>, f: set<int> -> a)
   requires s0 * s1 == {}
   ensures f((s0 + s1) * s0) == f(s0) {
     calc {
       (s0 + s1) * s0; (s0 * s0) + (s1 * s0);
       s0;
     }
   }\end{mcode}
  \caption{Proof in Dafny using built-in set axiomatization}
\end{subfigure}

\quad\\
\begin{subfigure}[b]{\textwidth}
  \begin{mcode}
{-@ assume unionEmpty :: ma : Set -> {v : () | ma \/ emptySet = ma } @-}
{-@ assume intersectComm :: ma : Set -> mb : Set -> {v : () | ma /\ mb = mb /\ ma } @-}
{-@ assume intersectSelf :: s0 : Set -> { s0 /\ s0 = s0 } @-}
{-@ assume unionIntersect :: s0 : Set -> s1 : Set -> s2 : Set ->
                           { (s0 \/ s1) /\ s2 = (s0 /\ s2) \/ (s1 /\ s2) } @-}
{-@ rwDisjoint :: s0 : Set -> {s1 : Set | IsDisjoint s0 s1} -> { s0 /\ s1 = emptySet } @-}

{-@ example1 :: s0 : Set -> { s1 : Set | IsDisjoint s0 s1 } ->  f : (Set -> a) ->
                { f ((s0 \/ s1) /\ s0) = f s0 } @-}
example1 s0 s1 _ = ()\end{mcode}
  \caption{An equivalent proof implemented in Liquid Haskell, with a user-defined axiomatization of sets.}
\end{subfigure}
\caption{Set Proofs in Liquid Haskell and Dafny}\label{fig:eval:dafny:set}
\end{figure}

In conclusion, we have shown that using \algoname to apply rewrites \nv{could}\nvout{can} be used as
an alternative to E-matching based axiomatization.
\nvtodo{I think this is a very strong sentence: from two examples we did not show we are an alternative to e mathing. replaced can with could to tone it down}
Furthermore, the termination
guarantee of \algoname enables axioms that may give rise to matching loops to,
instead, be encoded as rewrite rules.

\subsection{Simplification of Equational Proofs}
\label{sec:eval:simplification}
Finally, we evaluate how \algoname can simplify equational proofs.
We chose to include the set example from \cite{leino2013verified} (described in
\secref{eval:set}), data structure proofs from \cite{vazou_theorem_2018},
examples from the Liquid Haskell test suite, as well as our own case study. We
developed each example in Liquid Haskell both with and without rewriting, and
compared the timing and proof complexity.
Each proof using rewriting was evaluated using each different
ordering constraint algebras built-in to our Haskell \algoname~library.
The proofs in
\cite{vazou_theorem_2018} were selected because they require induction,
expansion of user-defined functions, and equational reasoning steps to prove
properties about trees and lists. The examples from the Liquid Haskell test
suite were taken to evaluate the rewriting across a range of representative
proofs. 

Our \texttt{DSL} case study evaluates the performance of our implementation using a \as{larger} set
of rewrite rules, by verifying optimizations for a simple programming language,
containing statements (\ie print, sequence, branches, repeats and no-ops) and
expressions (\ie constants, variables, arithmetic and boolean expressions)
using 23 rewrite rules. Our rewriting technique \as{can prove the} kind of
equivalences used in techniques such as
supercompilation~\cite{10.1145/2088456.1863540, WADLER1990231,eqsat-opt}, by
encoding the basic equality axioms as rewrite rules and using them to prove more
complicated theorems. A full list of the axioms and proved theorems are
available in \appref{eval:casestudy}. We note that we
encoded arithmetic operations as uninterpreted SMT functions, so that the
built-in arithmetic theory of the SMT does not aid proof automation.

\begin{table}
  \begin{tabular}{|l| c | c | c | c | c | c | c | c | c | c |}
    \hline
    \multirow{2}{*}{Name} &
    \multirow{2}{*}{Orig.} &
    \multirow{2}{*}{Cut} & \multirow{2}{*}{Rules} &\multicolumn{5}{c|}{Time}\\
    \cline{5-9}
    &&&&Orig.&RPQO&LPQO&KBQO&Fuel
    \\
    \hline
Set-Dafny      & 4        & 4       & 5    &  1.11s& \checkmark 1.15s& \checkmark 1.19s& \xmark 1.13s& \checkmark 1.22s\\
Set-Mono       & 7        & 7       & 4    &  1.16s& \xmark 1.40s& \xmark 1.41s& \checkmark 1.47s& \checkmark 1.60s\\
List           & 3        & 3       & 3    &  2.46s& \checkmark 3.17s& \xmark 4.21s& \xmark 2.24s& \checkmark 3.54s\\
Tree           & 3        & 3       & 3    &  1.61s& \checkmark 2.64s& \checkmark 3.40s& \checkmark 3.08s& \checkmark 3.12s\\
DSL            & 43       & 43      & 23   &  2.89s& \checkmark 5.46s& \xmark 3.85s& \xmark 4.19s& \checkmark 6.54s\\
LH-FingerTree  & 2        & 1       & 1    &  5.55s& \checkmark 5.60s& \checkmark 5.57s& \checkmark 5.64s& \checkmark 5.95s\\
LH-T1013       & 1        & 1       & 1    &  1.11s& \checkmark 1.06s& \checkmark 1.00s& \checkmark 1.02s& \checkmark 1.06s\\
LH-T1025       & 2        & 2       & 2    &  1.03s& \checkmark 1.05s& \checkmark 1.08s& \checkmark 1.07s& \checkmark 1.13s\\
LH-T1548       & 1        & 1       & 1    &  1.45s& \checkmark 1.33s& \checkmark 1.38s& \checkmark 1.32s& \checkmark 1.45s\\
LH-T1660       & 1        & 1       & 1    &  1.09s& \checkmark 1.12s& \checkmark 1.12s& \checkmark 1.12s& \checkmark 1.20s\\
LH-MapReduce   & 4        & 3       & 2    &  14.38s& \checkmark 29.50s& \checkmark 518.91s& \checkmark 28.49s& \xmark Timeout\\
    \hline
  \end{tabular}
  \caption{
    Results from simplification of proofs with rewriting. \textbf{Set-Dafny} is
    the set example from\cite{leino2013verified}, \textbf{Set-Mono} describes a
    similar property. \textbf{List} and
    \textbf{Tree} are equational proofs from \cite{vazou_theorem_2018}.
    \textbf{DSL} is the program equivalence case study. The remaining proofs are from the Liquid Haskell
    test suite folder \texttt{tests/pos}, excluding those using only inductive or mutually inductive lemmas.
    \textbf{Orig.} is the number of non-inductive lemma applications in the original proof.
    \textbf{Cut} is the number of lemma applications that were removed by rewriting.
    \textbf{Rules} is the number of axioms encoded as rewrite rules.
    \textbf{Time (Orig.)} is verification time in seconds for the original proof.
    \textbf{LPQO} and \textbf{KBQO} are OCAs derived from the Lexicographic Path
Ordering and Knuth-Bendix ordering respectively, and
    \textbf{Fuel} is an OCA allowing up to 5 rewrite applications per
    proof goal.
    }
  \label{table:evaluation:case-studies}
  \vspace{-0.3in}
\end{table}

We present our results in Table \ref{table:evaluation:case-studies}. By using
rewriting, we were able to eliminate all but two of the non-inductive axiom
instantiations, while maintaining a reasonable verification time. As expected,
no ordering constraint algebra was able to complete all the proofs using
rewriting; however, each proof could be verified with at least one of them.

The test cases !LH-FingerTree! and !LH-MapReduce! required manual axiom
instantiations because the structure of the term did not match the rewrite rule
for the axiom. \texttt{LH-MapReduce}, requires proving the identity \texttt{op
(f (take n is)) (mapReduce n f op (drop n is)) = f is}. An inductive lemma
application generates the background equality \texttt{mapReduce n f op (drop n
is) = f (drop n is)}, and a rewrite matching the term \texttt{op (f
(take n is)) (f (drop n is))} must be instantiated to complete the proof.
However, since the background equality is neither a rewrite rule nor an
evaluation step, the necessary term \texttt{op (f (take n is)) (f (drop n is))}
never appears. Therefore, it is necessary to manually instantiate the
lemma. As future work, a limited form of E-matching \cite{Ematching} could be
used to address this issue in the general case.

In conclusion, we've shown that extending Liquid Haskell to use \algoname
enables rewriting functionality not subsumed by existing theorem provers, that
\algoname is effective for axiom instantiation, and that \algoname can simplify
equational proofs.


\section{Related Work}
\label{sec:related}
\astodo{I guess we need to check whether this needs updating?}

\subparagraph*{Theorem Provers \& Rewriting}
Term rewriting is an effective technique to automate
theorem proving~\cite{hsiang_term_1992} supported by most standard theorem provers.
\S~\ref{sec:eval:comparison} compares, by examples, our technique
with \Coq, \Agda, \Lean, and \Isabelle.
%
In short, our approach is different because it uses user-specified rewrite rules
to derive, in a terminating way, equalities that strengthen the SMT-decidable
verification conditions generated during program verification.

\vspace{-0.3em}
\subparagraph*{SMT Verification \& Rewriting}
Our rewrite rules could be encoded in SMT solvers
as universally quantified equations and instantiated
using \textit{E-matching}~\cite{Ematching},
\ie a common algorithm for quantifier instantiation.
E-matching might generate matching loops leading to unpredictable divergence.
\cite{leino_trigger_2016} refers to this unpredictable behavior of E-matching
as the ``the butterfly effect'' and partially addresses it by detecting formulas that
could give rise to matching loops.
Our approach circumvents unpredictability by using the terminating \algoname
algorithm to instantiate the rewrite rules outside of the SMT solver.

Z3~\cite{Z3} and \CVCFour~\cite{CVC4} are state-of-the-art SMT solvers; both
support theory-specific rewrite rules internally. Recent work \cite{CVC4Rewrite}
enables user-provided rewrite rules to be added to \CVCFour. However,
using the SMT solver as a rewrite engine offers little control over rewrite rule
instantiation, which is necessary for ensuring termination.

\vspace{-0.3em}
\subparagraph*{Rewriting in Haskell}
Haskell itself has used various notions of rewriting for program verification.
GHC supports the !RULES! pragma with which the user can specify unchecked,
quantified expression equalities that are used at compile time for program optimization.
\cite{Breitner18} proposes Inspection Testing as a way to check such rewrite rules
using runtime execution and metaprogramming,
while~\cite{Hermit} prove rewrite rules via metaprogramming and user-provided hints.
In a work closely related to ours, Zeno~\cite{Zeno} is using rewriting,
induction, and further heuristics to provide lemma discovery and fully automatic
proof generation of inductive properties.
Unlike our approach, Zeno's syntax is restricted (\eg it does not allow for existentials)
and it does not allow for user-provided hints when automation fails.
HALO~\cite{vytiniotis2013halo} enables Haskell verification by converting
Haskell into logic and using an SMT solver
to verify user-defined formulas.
However, this approach relies on SMT quantifiers to encode user functions,
thus the solver can diverge and verification becomes unpredictable.

\vspace{-0.3em}
\subparagraph*{Termination of Rewriting and Runtime Termination Checking}
Early work on proving termination of rewriting using simplification orderings is
described in \cite{dershowitz1982orderings}. More recent work involves
dependency pairs \cite{arts_termination_2000} and applying the size-change
termination principle \cite{scp} in the context of rewriting
\cite{thiemann_size-change_2007}. Tools like AProVE \cite{aprove} and NaTT \cite{natt}
can statically prove the termination of rewriting.

In contrast, \algoname is not focused on statically proving termination of
rewriting; rather it uses a well-founded ordering to ensure termination at
runtime. This approach enables integration of arbitrary external oracles to
produce rewrite applications, as a static analysis is not possible in principle.
Furthermore, our approach enables nonterminating rewriting systems to be useful:
\algoname will still apply certain rewrite rules to satisfy a proof obligation,
even if the rewrite rules themselves cannot be statically shown to terminate.

We choose to use a well-quasi-ordering~\cite{kruskal_theory_1972} because it
enables rewriting to terms that are not strictly decreasing in a simplification
ordering. WQOs are commonly used in online termination
checking~\cite{goos_homeomorphic_2002}, especially for program optimization
techniques such as supercompilation~\cite{bolingbroke2011termination}.

\subparagraph*{Equality Saturation}

In our implementation, \algoname passes equalities to the SMT environment, ultimately used for \emph{equality saturation} via an E-graph data structure \cite{detlefs2005simplify}. Equality saturation has also been used for supercompilation\cite{eqsat-opt}. \algoname{} does not currently exploit equality saturation (unless indirectly via its oracle). However, as future work we might explore local usage of efficient E-graph implementations.
%
(\eg{} \cite{willsey2021egg}) for caching the
equivalence classes generated via rewrite applications.

\subparagraph*{Associative-Commutative Rewriting}

Traditionally, enforcing a strict ordering on terms prevents the application of
rewrite rules for associativity or commutativity (AC); this problem motivates
\algoname's use of well-quasi orders. However, another solution is to
omit the rules and instead perform the substitution step of rewriting modulo AC.
Termination of the resulting system can be proved using an AC ordering
\cite{dershowitz1983associative}; the essential requirement is that the ordering
also respects AC, such that $t > u$ implies $t' > u'$ for all terms $t'$
AC-equivalent to $t$ and $u'$ AC-equivalent to $u$.

\algoname's use of well-quasi-orderings enables AC axioms to be encoded as
rewrite rules, guaranteeing completeness if the AC-equivalence class of a term
is a subset of the equivalence class induced by the ordering. This is a
significant practical benefit as it does not require \algoname to identify AC symbols
and treat them differently for unification.

However, we note that treating AC axioms as rewrite rules can lead to an
explosion in the number of terms obtained via rewriting. As future work, it
could be possible to extend \algoname to support AC rewriting and unification in
order to reduce the number of explicitly instantiated terms.

\section{Conclusion}
\label{sec:conclusion}

We have presented \algoname, a novel approach to rewriting that uses an
online termination check that simultaneously considers entire families of term
orderings via a newly introduced Ordering Constraint Algebra.
We defined our algebra on well-quasi orderings that are more permissive than
standard simplification orderings, and demonstrated how to derive well-quasi
orderings from well-known simplification orderings.
In addition, we proved correctness, relative completeness, and (online)
termination of our algorithm.
\as{Our \algoname approach is designed, via a generic core algorithm and the pluggable abstraction of our OCAs, to be simple to (re-)implement and adapt to different programming languages or efficient implementations of term ordering families.}
We \as{demonstrated this by writing an implementation of} \algoname as a \as{small Haskell} library suitable for integration with
existing verification tools.
To evaluate \algoname we used our library to extend Liquid Haskell, and showed
that the resulting system \as{compares well with} \asout{is not subsumed by}
existing rewriting techniques, it can be used as an alternative to E-matching
based axiomatizations approaches \as{without risking non-termination}, and can \as{substantially} simplify \as{proofs requiring equational reasoning steps.}
%
%

\bibliography{bib}

\begin{thebibliography}{10}

\bibitem{AgdaUserManual}
{Agda Developers}.
\newblock {\em The Agda Language Reference, version 2.6.1}, 2020.
\newblock Available electronically at
  \url{https://agda.readthedocs.io/en/v2.6.1/language/index.html}.

\bibitem{amin2014computing}
Nada Amin, K~Rustan~M Leino, and Tiark Rompf.
\newblock Computing with an smt solver.
\newblock In {\em International Conference on Tests and Proofs}, pages 20--35.
  Springer, 2014.

\bibitem{arts_termination_2000}
Thomas Arts and Jürgen Giesl.
\newblock Termination of term rewriting using dependency pairs.
\newblock {\em Theoretical Computer Science}, 236(1):133--178, April 2000.
\newblock URL:
  \url{http://www.sciencedirect.com/science/article/pii/S0304397599002078},
  \href {https://doi.org/10.1016/S0304-3975(99)00207-8}
  {\path{doi:10.1016/S0304-3975(99)00207-8}}.

\bibitem{Lean_tp}
Jeremy Avigad, Leonardo de~Moura, and Soonho Kong.
\newblock {\em Theorem Proving in Lean, Release 3.20.0}, September 2020.
\newblock p 73.
\newblock URL:
  \url{https://leanprover.github.io/theorem_proving_in_lean/theorem_proving_in_lean.pdf}.

\bibitem{Lean}
Jeremy Avigad, Gabriel Ebner, and Sebastian Ullrich.
\newblock {\em The Lean Reference Manual, Release 3.3.0}, 2018.
\newblock URL: \url{https://leanprover.github.io/reference/lean_reference.pdf}.

\bibitem{CVC4}
Clark Barrett, Christopher~L. Conway, Morgan Deters, Liana Hadarean, Dejan
  Jovanovi{'{c}}, Tim King, Andrew Reynolds, and Cesare Tinelli.
\newblock {CVC4}.
\newblock In Ganesh Gopalakrishnan and Shaz Qadeer, editors, {\em Proceedings
  of the 23rd International Conference on Computer Aided Verification (CAV
  '11)}, volume 6806 of {\em Lecture Notes in Computer Science}, pages
  171--177. Springer, July 2011.
\newblock Snowbird, Utah.
\newblock URL: \url{http://www.cs.stanford.edu/~barrett/pubs/BCD+11.pdf}.

\bibitem{BeckerMuellerSummers19}
N.~Becker, P.~M\"uller, and A.~J. Summers.
\newblock The axiom profiler: Understanding and debugging smt quantifier
  instantiations.
\newblock In {\em Tools and Algorithms for the Construction and Analysis of
  Systems (TACAS) 2019}, LNCS, pages 99--116. Springer-Verlag, 2019.

\bibitem{10.1145/2088456.1863540}
Maximilian Bolingbroke and Simon Peyton~Jones.
\newblock Supercompilation by evaluation.
\newblock {\em SIGPLAN Not.}, 45(11):135–146, September 2010.
\newblock \href {https://doi.org/10.1145/2088456.1863540}
  {\path{doi:10.1145/2088456.1863540}}.

\bibitem{bolingbroke2011termination}
Maximilian Bolingbroke, Simon Peyton~Jones, and Dimitrios Vytiniotis.
\newblock Termination combinators forever.
\newblock In {\em Proceedings of the 4th ACM symposium on Haskell}, pages
  23--34, 2011.

\bibitem{Breitner18}
Joachim Breitner.
\newblock A promise checked is a promise kept: inspection testing.
\newblock In Nicolas Wu, editor, {\em Proceedings of the 11th {ACM} {SIGPLAN}
  International Symposium on Haskell, Haskell@ICFP 2018, St. Louis, MO, USA,
  September 27-17, 2018}, pages 14--25. {ACM}, 2018.
\newblock \href {https://doi.org/10.1145/3242744.3242748}
  {\path{doi:10.1145/3242744.3242748}}.

\bibitem{Coq}
The {Coq} {Development}~{Team}.
\newblock {\em The {Coq} Reference Manual, version 8.11.2}, 2020.
\newblock Available electronically at \url{http://coq.inria.fr/refman}.

\bibitem{Ematching}
Leonardo de~Moura and Nikolaj Bj{\o}rner.
\newblock Efficient e-matching for smt solvers.
\newblock In Frank Pfenning, editor, {\em Automated Deduction -- CADE-21},
  pages 183--198, Berlin, Heidelberg, 2007. Springer Berlin Heidelberg.

\bibitem{Z3}
Leonardo De~Moura and Nikolaj Bj{\o}rner.
\newblock Z3: An efficient smt solver.
\newblock In {\em International conference on Tools and Algorithms for the
  Construction and Analysis of Systems}, pages 337--340. Springer, 2008.

\bibitem{dershowitz_simplification_1979}
Nachum Dershowitz.
\newblock A note on simplification orderings.
\newblock {\em Information Processing Letters}, 9(5):212--215, 1979.
\newblock \href {https://doi.org/https://doi.org/10.1016/0020-0190(79)90071-1}
  {\path{doi:https://doi.org/10.1016/0020-0190(79)90071-1}}.

\bibitem{dershowitz1982orderings}
Nachum Dershowitz.
\newblock Orderings for term-rewriting systems.
\newblock {\em Theoretical computer science}, 17(3):279--301, 1982.

\bibitem{dershowitz1987termination}
Nachum Dershowitz.
\newblock Termination of rewriting.
\newblock {\em Journal of symbolic computation}, 3(1-2):69--115, 1987.

\bibitem{dershowitz1983associative}
Nachum Dershowitz, Jieh Hsiang, N~Alan Josephson, and David~A Plaisted.
\newblock Associative-commutative rewriting.
\newblock In {\em IJCAI}, pages 940--944, 1983.

\bibitem{dershowitz_proving_1979}
Nachum Dershowitz and Zohar Manna.
\newblock Proving termination with multiset orderings.
\newblock In Hermann~A. Maurer, editor, {\em Automata, {Languages} and
  {Programming}}, Lecture {Notes} in {Computer} {Science}, pages 188--202,
  Berlin, Heidelberg, 1979. Springer.
\newblock \href {https://doi.org/10.1007/3-540-09510-1_15}
  {\path{doi:10.1007/3-540-09510-1_15}}.

\bibitem{Detlefs05}
David Detlefs, Greg Nelson, and James~B. Saxe.
\newblock Simplify: A theorem prover for program checking.
\newblock {\em J. ACM}, 52(3):365--473, May 2005.
\newblock URL: \url{http://doi.acm.org/10.1145/1066100.1066102}, \href
  {https://doi.org/10.1145/1066100.1066102}
  {\path{doi:10.1145/1066100.1066102}}.

\bibitem{detlefs2005simplify}
David Detlefs, Greg Nelson, and James~B Saxe.
\newblock Simplify: a theorem prover for program checking.
\newblock {\em Journal of the ACM (JACM)}, 52(3):365--473, 2005.

\bibitem{dross2016adding}
Claire Dross, Sylvain Conchon, Johannes Kanig, and Andrei Paskevich.
\newblock Adding decision procedures to smt solvers using axioms with triggers.
\newblock {\em Journal of Automated Reasoning}, 56(4):387--457, 2016.

\bibitem{Hermit}
Andrew Farmer, Neil Sculthorpe, and Andy Gill.
\newblock Reasoning with the hermit: Tool support for equational reasoning on
  ghc core programs.
\newblock In {\em Proceedings of the 2015 ACM SIGPLAN Symposium on Haskell},
  Haskell ’15, page 23–34, New York, NY, USA, 2015. Association for
  Computing Machinery.
\newblock \href {https://doi.org/10.1145/2804302.2804303}
  {\path{doi:10.1145/2804302.2804303}}.

\bibitem{FilliatreP13}
Jean-Christophe Filli{\^a}tre and Andrei Paskevich.
\newblock Why3---where programs meet provers.
\newblock In Matthias Felleisen and Philippa Gardner, editors, {\em Programming
  Languages and Systems (ESOP)}, volume 7792 of {\em Lecture Notes in Computer
  Science}, pages 125--128. Springer, 2013.

\bibitem{aprove}
J{\"u}rgen Giesl, Cornelius Aschermann, Marc Brockschmidt, Fabian Emmes,
  Florian Frohn, Carsten Fuhs, Jera Hensel, Carsten Otto, Martin Pl{\"u}cker,
  Peter Schneider-Kamp, et~al.
\newblock Analyzing program termination and complexity automatically with
  aprove.
\newblock {\em Journal of Automated Reasoning}, 58(1):3--31, 2017.

\bibitem{hsiang_term_1992}
Jieh Hsiang, Hélène Kirchner, Pierre Lescanne, and Michaël Rusinowitch.
\newblock The term rewriting approach to automated theorem proving.
\newblock {\em The Journal of Logic Programming}, 14(1):71--99, October 1992.
\newblock URL:
  \url{http://www.sciencedirect.com/science/article/pii/0743106692900477},
  \href {https://doi.org/10.1016/0743-1066(92)90047-7}
  {\path{doi:10.1016/0743-1066(92)90047-7}}.

\bibitem{Huet77}
Gerard Huet.
\newblock Confluent reductions: Abstract properties and applications to term
  rewriting systems.
\newblock In {\em Proceedings of the 18th Annual Symposium on Foundations of
  Computer Science}, SFCS '77, page 30–45, USA, 1977. IEEE Computer Society.
\newblock \href {https://doi.org/10.1109/SFCS.1977.9}
  {\path{doi:10.1109/SFCS.1977.9}}.

\bibitem{numwqos}
OEIS~Foundation Inc.
\newblock The {O}n-{L}ine {E}ncyclopedia of {I}nteger {S}equences, 2022.
\newblock A000798: Number of different quasi-orders (or topologies, or
  transitive digraphs) with n labeled elements.
\newblock URL: \url{https://oeis.org/A000798}.

\bibitem{kaplan1984}
Stéphane Kaplan.
\newblock Conditional rewrite rules.
\newblock {\em Theoretical Computer Science}, 33(2):175--193, 1984.
\newblock URL:
  \url{https://www.sciencedirect.com/science/article/pii/0304397584900872},
  \href {https://doi.org/https://doi.org/10.1016/0304-3975(84)90087-2}
  {\path{doi:https://doi.org/10.1016/0304-3975(84)90087-2}}.

\bibitem{KlopTRS}
J.~W. Klop.
\newblock {\em Term Rewriting Systems}, page 1–116.
\newblock Oxford University Press, Inc., USA, 1993.

\bibitem{kbo}
Donald~E Knuth and Peter~B Bendix.
\newblock Simple word problems in universal algebras.
\newblock In {\em Automation of Reasoning}, pages 342--376. Springer, 1983.

\bibitem{kruskal_theory_1972}
Joseph~B Kruskal.
\newblock The theory of well-quasi-ordering: {A} frequently discovered concept.
\newblock {\em Journal of Combinatorial Theory, Series A}, 13(3):297--305,
  November 1972.
\newblock URL:
  \url{http://www.sciencedirect.com/science/article/pii/0097316572900635},
  \href {https://doi.org/10.1016/0097-3165(72)90063-5}
  {\path{doi:10.1016/0097-3165(72)90063-5}}.

\bibitem{scp}
Chin~Soon Lee, Neil~D. Jones, and Amir~M. Ben-Amram.
\newblock The size-change principle for program termination.
\newblock In {\em Proceedings of the 28th ACM SIGPLAN-SIGACT Symposium on
  Principles of Programming Languages}, POPL ’01, page 81–92, New York, NY,
  USA, 2001. Association for Computing Machinery.
\newblock \href {https://doi.org/10.1145/360204.360210}
  {\path{doi:10.1145/360204.360210}}.

\bibitem{leino_trigger_2016}
K.~R.~M. Leino and Clément Pit-Claudel.
\newblock Trigger {Selection} {Strategies} to {Stabilize} {Program}
  {Verifiers}.
\newblock In Swarat Chaudhuri and Azadeh Farzan, editors, {\em Computer {Aided}
  {Verification}}, Lecture {Notes} in {Computer} {Science}, pages 361--381,
  Cham, 2016. Springer International Publishing.
\newblock \href {https://doi.org/10.1007/978-3-319-41528-4_20}
  {\path{doi:10.1007/978-3-319-41528-4_20}}.

\bibitem{dafny}
K.~Rustan~M. Leino.
\newblock Dafny: An automatic program verifier for functional correctness.
\newblock In {\em Proceedings of the 16th International Conference on Logic for
  Programming, Artificial Intelligence, and Reasoning}, LPAR’10, page
  348–370, Berlin, Heidelberg, 2010. Springer-Verlag.

\bibitem{leino2013verified}
K~Rustan~M Leino and Nadia Polikarpova.
\newblock Verified calculations.
\newblock In {\em Working Conference on Verified Software: Theories, Tools, and
  Experiments}, pages 170--190. Springer, 2013.

\bibitem{goos_homeomorphic_2002}
Michael Leuschel.
\newblock Homeomorphic {Embedding} for {Online} {Termination} of {Symbolic}
  {Methods}.
\newblock In Gerhard Goos, Juris Hartmanis, Jan van Leeuwen, Torben~Æ.
  Mogensen, David~A. Schmidt, and I.~Hal Sudborough, editors, {\em The
  {Essence} of {Computation}}, volume 2566, pages 379--403. Springer Berlin
  Heidelberg, Berlin, Heidelberg, 2002.
\newblock Series Title: Lecture Notes in Computer Science.
\newblock URL: \url{http://link.springer.com/10.1007/3-540-36377-7_17}, \href
  {https://doi.org/10.1007/3-540-36377-7_17}
  {\path{doi:10.1007/3-540-36377-7_17}}.

\bibitem{meng2008translating}
Jia Meng and Lawrence~C Paulson.
\newblock Translating higher-order clauses to first-order clauses.
\newblock {\em Journal of Automated Reasoning}, 40(1):35--60, 2008.

\bibitem{MuellerSchwerhoffSummers16}
P.~M{\"u}ller, M.~Schwerhoff, and A.~J. Summers.
\newblock Viper: A verification infrastructure for permission-based reasoning.
\newblock In {\em VMCAI}, 2016.

\bibitem{Isabelle}
Tobias Nipkow, Markus Wenzel, and Lawrence~C. Paulson.
\newblock {\em Isabelle/HOL: A Proof Assistant for Higher-Order Logic}.
\newblock Springer-Verlag, 2020.

\bibitem{Agda}
Ulf Norell.
\newblock Dependently typed programming in agda.
\newblock In {\em Proceedings of the 6th International Conference on Advanced
  Functional Programming}, AFP’08, page 230–266, Berlin, Heidelberg, 2008.
  Springer-Verlag.

\bibitem{CVC4Rewrite}
Andres N{\"o}tzli, Andrew Reynolds, Haniel Barbosa, Aina Niemetz, Mathias
  Preiner, Clark Barrett, and Cesare Tinelli.
\newblock Syntax-guided rewrite rule enumeration for smt solvers.
\newblock In Mikol{\'a}{\v{s}} Janota and In{\^e}s Lynce, editors, {\em Theory
  and Applications of Satisfiability Testing -- SAT 2019}, pages 279--297,
  Cham, 2019. Springer International Publishing.

\bibitem{paulson2007source}
Lawrence~C Paulson and Kong~Woei Susanto.
\newblock Source-level proof reconstruction for interactive theorem proving.
\newblock In {\em International Conference on Theorem Proving in Higher Order
  Logics}, pages 232--245. Springer, 2007.

\bibitem{paulsson2012three}
Lawrence~C Paulsson and Jasmin~C Blanchette.
\newblock Three years of experience with sledgehammer, a practical link between
  automatic and interactive theorem provers.
\newblock In {\em Proceedings of the 8th International Workshop on the
  Implementation of Logics (IWIL-2010), Yogyakarta, Indonesia. EPiC}, volume~2,
  2012.

\bibitem{framac}
Julien Signoles, Pascal Cuoq, Florent Kirchner, Nikolai Kosmatov, Virgile
  Prevosto, and Boris Yakobowski.
\newblock Frama-c: a software analysis perspective.
\newblock volume~27, 10 2012.
\newblock \href {https://doi.org/10.1007/s00165-014-0326-7}
  {\path{doi:10.1007/s00165-014-0326-7}}.

\bibitem{Zeno}
William Sonnex, Sophia Drossopoulou, and Susan Eisenbach.
\newblock Zeno: An automated prover for properties of recursive data
  structures.
\newblock In Cormac Flanagan and Barbara K{\"o}nig, editors, {\em Tools and
  Algorithms for the Construction and Analysis of Systems}, pages 407--421,
  Berlin, Heidelberg, 2012. Springer Berlin Heidelberg.

\bibitem{FStar}
Nikhil Swamy, Catalin Hritcu, Chantal Keller, Aseem Rastogi, Antoine
  Delignat-Lavaud, Simon Forest, Karthikeyan Bhargavan, C\'{e}dric Fournet,
  Pierre-Yves Strub, Markulf Kohlweiss, Jean-Karim Zinzindohou\'e, and Santiago
  {Zanella-B\'eguelin}.
\newblock Dependent types and multi-monadic effects in {F*}.
\newblock In {\em 43rd ACM SIGPLAN-SIGACT Symposium on Principles of
  Programming Languages (POPL)}, pages 256--270. ACM, January 2016.
\newblock URL: \url{https://www.fstar-lang.org/papers/mumon/}.

\bibitem{eqsat-opt}
Ross Tate, Michael Stepp, Zachary Tatlock, and Sorin Lerner.
\newblock Equality saturation: a new approach to optimization.
\newblock In {\em {POPL} '09: Proceedings of the 36th annual {ACM}
  {SIGPLAN-SIGACT} symposium on Principles of Programming Languages}, pages
  264--276, New York, NY, USA, 2009. {ACM}.
\newblock URL: \url{http://www.cs.cornell.edu/~ross/publications/eqsat/}, \href
  {https://doi.org/http://doi.acm.org/10.1145/1480881.1480915}
  {\path{doi:http://doi.acm.org/10.1145/1480881.1480915}}.

\bibitem{thiemann_size-change_2007}
René Thiemann and Jürgen Giesl.
\newblock Size-{Change} {Termination} for {Term} {Rewriting}.
\newblock volume 2706, pages 264--278, March 2007.
\newblock \href {https://doi.org/10.1007/3-540-44881-0_19}
  {\path{doi:10.1007/3-540-44881-0_19}}.

\bibitem{vazou_theorem_2018}
Niki Vazou, Joachim Breitner, Rose Kunkel, David Van~Horn, and Graham Hutton.
\newblock Theorem proving for all: equational reasoning in liquid {Haskell}
  (functional pearl).
\newblock In {\em Proceedings of the 11th {ACM} {SIGPLAN} {International}
  {Symposium} on {Haskell}}, Haskell 2018, pages 132--144, St. Louis, MO, USA,
  September 2018. Association for Computing Machinery.
\newblock \href {https://doi.org/10.1145/3242744.3242756}
  {\path{doi:10.1145/3242744.3242756}}.

\bibitem{LiquidHaskell}
Niki Vazou, Eric~L. Seidel, Ranjit Jhala, Dimitrios Vytiniotis, and Simon
  Peyton-Jones.
\newblock Refinement types for haskell.
\newblock In {\em Proceedings of the 19th ACM SIGPLAN International Conference
  on Functional Programming}, ICFP ’14, page 269–282, New York, NY, USA,
  2014. Association for Computing Machinery.
\newblock \href {https://doi.org/10.1145/2628136.2628161}
  {\path{doi:10.1145/2628136.2628161}}.

\bibitem{PLE}
Niki Vazou, Anish Tondwalkar, Vikraman Choudhury, Ryan~G. Scott, Ryan~R.
  Newton, Philip Wadler, and Ranjit Jhala.
\newblock Refinement reflection: Complete verification with smt.
\newblock {\em Proc. ACM Program. Lang.}, 2(POPL), December 2017.
\newblock \href {https://doi.org/10.1145/3158141} {\path{doi:10.1145/3158141}}.

\bibitem{vytiniotis2013halo}
Dimitrios Vytiniotis, Simon Peyton~Jones, Koen Claessen, and Dan Ros{\'e}n.
\newblock Halo: Haskell to logic through denotational semantics.
\newblock In {\em Proceedings of the 40th annual ACM SIGPLAN-SIGACT symposium
  on Principles of programming languages}, pages 431--442, 2013.

\bibitem{WADLER1990231}
Philip Wadler.
\newblock Deforestation: transforming programs to eliminate trees.
\newblock {\em Theoretical Computer Science}, 73(2):231 -- 248, 1990.
\newblock URL:
  \url{http://www.sciencedirect.com/science/article/pii/030439759090147A},
  \href {https://doi.org/https://doi.org/10.1016/0304-3975(90)90147-A}
  {\path{doi:https://doi.org/10.1016/0304-3975(90)90147-A}}.

\bibitem{willsey2021egg}
Max Willsey, Chandrakana Nandi, Yisu~Remy Wang, Oliver Flatt, Zachary Tatlock,
  and Pavel Panchekha.
\newblock Egg: Fast and extensible equality saturation.
\newblock {\em Proceedings of the ACM on Programming Languages}, 5(POPL):1--29,
  2021.

\bibitem{natt}
Akihisa Yamada, Keiichirou Kusakari, and Toshiki Sakabe.
\newblock Nagoya termination tool.
\newblock In {\em Rewriting and Typed Lambda Calculi}, pages 466--475.
  Springer, 2014.

\end{thebibliography}
\clearpage
\appendix
\section{Metaproperty Proofs}\label{appendix:metaprop:proofs}

\pathinvariant*
\begin{proof}
    By induction on the loop iterations of the algorithm.
    $p$ is initialized with the single element $([\tt_{0}], c)$. $[\tt_0]$ is a
  valid path of $\rrel + \evalSymb$, because it only contains a single term;
  clearly this path also starts with $\tt_0$.

    At each loop iteration, new elements are potentially pushed to $p$.
    Suppose the path $ts$ is popped from $p$ at the beginning of the loop.
    The element to be pushed is a pair $(ts \concat [t'], c)$ where
  $\goestor{\rrel + \evalSymb}{last(ts)}{\tt'}$. This exactly satisfies the
  inductive hypothesis: if $ts$ is a path of $\rrel + \evalSymb$, then
  $ts \concat [\tt']$ is also a path of $\rrel + \evalSymb$. Furthermore, this
  operation preserves the head of the list: $t_0$ is still the first element.
\end{proof}

\relcompleteness*
\begin{proof}
  The proof goes by induction on the number of steps of the path.

  Assume the path has $n$ steps:
  $\goestor{\evalSymb}{\tt_0}{}\goestor{\evalSymb}{\tt_1}{} \dots \goestor{\evalSymb}{}{} \goestor{\evalSymb}{\tt_{n-1}}{\tt_n} \equiv u$.

  For the base case, $n=0$ and $u \equiv \tt_0$.
  Since $p$ is initialized with $([\tt_0], \octop)$, by the Invariant~\ref{invariant:output},
  $\tt \in \algod{\tt_{0}}$.

  For the inductive case, assume that
  $\evalstor{\evalSymb}{\tt_0}{\tt_{n-1}}\goestor{\evalSymb}{}{\tt_{n}}$.
  By inductive hypothesis, $\tt_{n-1} \in \algod{\tt_{0}}$.
  When $\tt_{n-1}$ was added in the result, it was the last
  element of a path $ts$ that was popped from the stack $p$.
  Since \goestor{\evalSymb}{\tt_{n-1}}{\tt_n}, we split cases
  on whether or not $\tt_n \in ts$.
  If $\tt_n \in ts$, then by Invariant~\ref{invariant:output} $\tt_n \in \algod{\tt_{0}}$.
  Otherwise, $(ts \concat [\tt_n],c)$ will be pushed into $p$ and, again,
  by Invariant~\ref{invariant:output} it will appear in the output.
\end{proof}

\relcompletenesstwo*

First, we observe the (somewhat standard) property that if any path justifies $\evalstor{\rrel + \evalSymb}{\tt_0}{\uu}$, there is a \emph{duplicate-free variant} of such path (intuitively, obtained by cutting out all subpaths leading from a term to itself).

  Below, we prove that
  if \evalstor{\rrel + \evalSymb}{\tt_0}{\uu} and the ordering \ord orients the path,
  then a duplicate-free variant path $ts$ belongs in the stack $p$ with some
  constraints $c$  and  $\ord~\in \gamma(c)$.

  \begin{invariant}\label{invariant:three}
    For any execution of \algod{\tt_0},
    if \evalstor{\rrel + \evalSymb}{\tt_0}{\tt_n}
    and $\ord~\in\gamma(\octop)$ is an ordering that orients
    \evalstor{\rrel + \evalSymb}{\tt_0}{\uu}, then at some iteration of the main loop,
    a duplicate-free variant path $ts$ of this path is stored in $p$,
    with some ordering constraints $c$ and $\ord~\in \gamma(c)$.
  \end{invariant}
  \begin{proof}
    The proof goes by strong induction on the length $n+1$ of the path justifying $\evalstor{\rrel + \evalSymb}{\tt_0}{\tt_n}$.

    First, consider the case $n = 0$, where the path is $[\tt_0]$ and the constraints
    $\octop$. $([\tt_0], \octop) \in p$ by initialization and trivially $\ord~\in \gamma(\octop)$.

    Otherwise, when $n > 0$, assume that
    $\evalstor{\rrel + \evalSymb}{\tt_0}{\tt_{n-1}}\goestor{\rrel + \evalSymb}{}{\tt_{n}}$.
    If there are any duplicate terms in this path, a duplicate-free variant exists of shorter length, and we can conclude by our induction hypothesis. Otherwise, consider this path with the last element $\tt_{n}$ removed. Being already duplicate-free, by our induction hypothesis we must have that, at some iteration of our main loop, this path is contained in $p$ along with a constraint $c_{n-1}$ such that $\ord~\in \gamma(c_{n-1})$. By the assumption that \ord orients the original path, in particular we must have $\tt_{n-1}\ord\tt_n$, and
    so, by \defref{oca},
    $\ord~\in \gamma(\ocrefine{c_{n-1}}{\tt}{\tt'})$ and therefore $\ocrefine{c_{n-1}}{\tt}{\tt'}$ is satisfiable. Therefore, the original path will be pushed to $p$ with this constraint in this loop iteration.
  \end{proof}

\begin{proof}
  The proof is similar to Theorem~\ref{thm:rewrite_complete_eval},
  but now we need to also show that the relation that orients the path
  satisfies all the ordering constraints generated by the respective \algoname path. By Invariant~\ref{invariant:three}, at some iteration of the main loop, there must be some path ending in $\uu$ contained in $p$. Then, by Invariant~\ref{invariant:output} it follows that all the elements of the path,
  thus also \uu, belong in the result.

\end{proof}

\thmterminating*

\begin{proof}

  At every iteration of $\algoname$, a path with length $n$ is popped off the
stack and due to Requirement~\ref{tass:eval}, and the fact that only a finite
number of new terms can be generated by single applications of the rules \rrel
to an arbitrary term, a finite number of paths with length $n+1$ is pushed on.
Therefore, $\algoname$ implicitly builds (via its set of paths $p$) a \emph{finitely-branching} tree starting from $t_{0}$.
For $\algoname$ to not terminate, there must be an infinite path down the tree (note that Requirement \ref{tass:funcs} eliminates the possibility that the operations called from the \oca{} cause non-termination).

Consider an arbitrary path down the tree explored by $\algoname$, represented by the $(\ts,\ordconstraints)$ pairs iteratively generated. Firstly, due to the first condition in the \texttt{foreach} of $\algoname$ (\cf{} \figref{algo-rewrite}), this path will remain duplicate-free.
 By Requirement \ref{tass:wf}, at only finitely many steps is the constraint
 tracked \emph{strictly} refined. Consider then, the postfix of the path after
 the last time that this happens; at every step, the constraint
 $\ordconstraints$ remains identical. The normalization assumption (Requirement
 \ref{tass:eval}) of \evalSymb entails that this path contains no infinite
 sequence of steps all justified by $\evalSymb$. However, for each step justified
 instead by a rewriting step from $\rrel$, the additional condition $\sat{c}$
 must hold; by \defref{oca} this means that there is some $\ord~\in\gamma(c)$
 which orients all of these steps. As $\ord$ must be an instance of a term
 ordering family (\defref{tofamily}), it is a thin well-quasi-order.
 Therefore $\ord$ can only orient a finite number of steps, and the path down
 the tree must be finite.

  Since every path in the finitely-branching tree explored is finite, the algorithm (always) terminates.
\end{proof}

\clearpage

\section{Proofs on Orderings}
\label{appendix:ordering:proofs}

\begin{lemma}\label{lemma:mul_subset}
  If $T \mul U$, then $T \mul U'$ for all $U' \subset U$.
\end{lemma}

\begin{proof}
    It is sufficient to show that $T \mul U$ implies $T \mul (U - u')$, for any
  $u' \in U$, since the subset can be obtained by removing a finite number of
  elements. That is, if $U'$ was obtained by removing elements \range{u}{n} from
  $U$, we
  can show that $T \mul (U \setminus \{u_{1}\})$ implies
  $T \mul (U \setminus \{u_{1}, u_{2}\})$ and so on.

  The proof goes by induction on the size of $T$ and case analysis on $T \mul U$.

  For case one there are no $u'$ in $U$, so the proof holds vacuously.

  For case two, we have either $u = u'$ or $u \neq u'$. If $u = u'$, a proof of
$T \mul (U - u)$ can be made by modifying the proof of $(T - t) \mul (U - u)$.
The base case of that proof must be of the form $T' \mul \emptyset$. We modify
the base case to be $(T' + t) \mul \emptyset$. Each recursive case is also
modified to replace $T'$ with $(T' + t)$, yielding $(T' + t) \mul (U - u)$ =
$T \mul (U - u)$, as required. The proof that $T \mul (U - u')$ for all other
$u' \in U$ is obtained by induction. By the inductive hypothesis, we have
$(T - t) \mul (U - u - u')$,  since $u \neq u'$, we also
have $u \in (U - u')$. Therefore applying case two we get $T \mul (U - u')$.

  For case three, we have either $u' < t$ or $u' \nless t$. If $u' < t$, then
the proof $(T - t) \mul (U \setminus \setfilter{u}{U}{u < t})$ is also a proof
of $(T - t) \mul ((U - u') \setminus \setfilter{u}{U}{u < t})$, thus we obtain
obtain the proof directly. The proof for all other $u' \in U$ is obtained by
induction. By the inductive hypothesis we have
$(T - t) \mul ((\mulsthreerhs{u}{U}{t}) - u')$.
Then, applying the same top-level proof yields $T \mul (U - u')$, since
$u'$ is not in the set \setfilter{u}{U}{u < t}.

\end{proof}

\begin{lemma}\label{lemma:mul_qo}
If $\ord_{X}$ is a quasi-order,
then the multiset extension $\mul$ is also a quasi-order.
\end{lemma}
\begin{proof}
        To show that $\mul$ is a quasi-order, we define a single-step version
  $\mulstep$, and show that $T \mul U$ if and only if $T \mulstar U$, where
  $\mulstar$ is the reflexive transitive closure of $U$.

  We define $\mulstep$ as:
  \begin{enumerate}
    \item For all elements $t, u$ if $t \in T$ and $u \approx t$, then $T \mulstep (T - t + u)$
    \item For all elements $t \in T$ and finite multisets $U$, if $t > u$ for all
          $u \in U$, then $T \mulstep ((T - t) \cup U)$
  \end{enumerate}

      First, observe that $\mulstar$ is monotonic with respect to multiset union:
  for all multisets $T$, $U$, and $V$, $T \mulstar U$ implies
  $(T \cup V) \mulstar (U \cup V)$.

  The reflexive case is given by $T \cup V = T \cup V$; we show the transitive
  case by showing there is a correspondence for each single-step. The proof
  for each case assumes an arbitrary multiset $V$.

  In case one we must show for all $t, u \in T$, $T \mulstar (T - t + u)$ implies
  $(T \cup V) \mulstar ((T - t + u) \cup V)$. $t$ and $u$ are also in $T \cup V$,
  therefore we have $(T \cup V) \mulstar ((T \cup V) - t + u)$. We have
  $(T \cup V) - t + u  =  (T - t + u) \cup V$, giving us the desired result.
  Case two is similar: $t \in T$ implies $t \in (T \cup V)$, and
  $((T \cup V) - t) \cup U = ((T -t) \cup U) \cup V$ for all $U$, $V$.

  Now we show the if direction by case analysis.
  \smallskip
  \begin{enumerate}[label=]
    \item Case 1: $U = \emptyset$.
    \\If $T = \emptyset$, then we have $T \mulstar U$ via
    reflexivity. Otherwise we can select an arbitrary $t$ to remove from $T$, and
    by definition of $\mulstep$ we have $T \mulstep ((T - t) \cup \emptyset)$.
    Then by induction on the size of $T$ we have
    $((T - t) \cup \emptyset) \mulstar \emptyset$. Then
    $T \mulstep ((T - t) \cup \emptyset) \mulstar \emptyset$, as required.
    \\
    \item Case 2:
    $t \in T \land u \in U \land t \approx u \land (T - t) \geqslant_{mul} (U - u)$.
    \\
    Let $T' = T - t$ and $U' = U - u$.
          Then we have
          $(T' + t) \mulstep (T' + u)$ by definition and
          $T' \mul U'$ implies $T' + u \mulstar U' + u$ via the inductive
          hypothesis and monotonicity.
          Thus
          $T = (T' + t) \mulstep (T' + u) \mulstar (U' + u) = U$ as required.
    \\
    \item Case 3: $t \in T \land (T - t) \geqslant_{mul} (U \setminus \{ u \in U~|~u < t \})$\\
          Partition $U$ into two sets $U_1$ and $U_2$ where
          $U_{1} = \{u \in U~|~ u \not< t\}$ and
          $U_{2} = \{u \in U~|~ u < t\}$.
          By definition we have $U = U_{1} \cup U_{2}$.
          As before $T' = T - t$. Then we have
$(T' + t) \mulstep (T' \cup U_{2})$. $T' \mul U_{1}$ implies
$(T' \cup U_{2}) \mulstar (U_{1} \cup U_{2})$ via monotonicity and induction.
          Thus
          $T = (T' + t) \mulstep (T' \cup U_{2}) \mulstar (U_{1} \cup U_{2}) = U$
          as required.
  \end{enumerate}

Now the only-if direction. First we have that $\mul$ is reflexive via
induction on size with base case $T = U = \emptyset$ handled by case 1, and
recursive case by case 2, similar to above, we remove an arbitrary $t$ from
$T$. Now we show how to handle one or more steps from $\mulstep$ in a single
step of $\mul$.

The key observation is that all elements $u$ of $U$, have exactly one
``responsible'' element $t$ in $T$ that justifies $T \mulstar U$: we must have
either $t > u$ or $t \approx u$ (in which case $t$ is uniquely responsible for
$u$ and no other elements of $U$). To prove $T \mul U$, for each $t$ in $T$, we
recursively build a tuple $(T', U', p)$ where $T'$, and $U'$ are multisets and
$p$ is the proof that $T' \mul U'$. The tuple is initialized to
$(\emptyset, \emptyset, U = \emptyset)$.

For each $t$ uniquely responsible for one $u$, we update the tuple to
$(T' + t, U' + u, t \in T \land u \in U \land t \approx u \land p)$. The new
proof state is valid because by induction we have $p$ being a proof of
$T' \mul U'$, as required.

Now consider each $t \in T$ where $t$ justified some multiset $U''$.
By induction, we have a proof of $T' \mul U'$; we need a proof that
$T' \mul ((U' \cup U'') \setminus \setfilter{u}{(U' \cup U'')}{u < t})$.
Since we have $t > u$ for all $u \in U''$, this simplifies to:
$T' \mul (U' \setminus \setfilter{u}{U'}{u < t})$, which we can obtain via the
hypothesis $T' \mul U'$ and lemma \ref{lemma:mul_subset}.

\end{proof}

\begin{lemma}\label{lemma:mul_wf}
  If $\ord_{X}$ is a well-quasi-order, the strict part of it's multiset
extension defined as $t \muls u$ if $t \mul u $ and $u \not\mul t$ is
a well-founded order.
\end{lemma}
\begin{proof}

    This proof operates on the single-step relation defined in \ref{lemma:mul_qo}.
    Proving the well-founded property is done by showing that an infinite
    descent in $\muls$ would correspond to an infinite descent in the underlying ordering.

    Now, consider a tree built from an infinite path $T_{1}, T_{2}, \ldots$
    of multisets related by $\mul$. With the exception of special nodes $\top$
    and $\bot$, each node in the tree represents an element in a multiset, and
    the vertices connect the elements to the smaller ones they were replaced with
    via an application of $\mulstep$. Crucially, every edge represents an descent in
    a well-founded order.

    The tree is constructed as follows: let $\top$ be the root of the tree, and
    let the elements of $T_{1}$ be the children of $\top$. Then, for each $T_{i}$ in
    the infinite list, it was either obtained by replacing some element in $T_{i-1}$
    with a same-sized element, or by removing some element $t$ and replacing it with
    a finite number of smaller elements \textit{ts}.

    In the former case, the tree is not modified.

    In the latter case, if $\textit{ts} = \emptyset$, add a single child $\bot$ to the
    $t$ in the tree. Otherwise, let \textit{ts} be the children of $t$.

        Now, we note that the case one of $\mulstep$ is symmetric. Therefore, each
    pair of terms related by $\muls$ must correspond to at least one step in case
    two of $\mulstep$, Therefore in an infinite path of terms related by $\muls$
    contains an infinite number of applications of case two in $\mulstep$.

    Therefore, an infinite number of vertices will be added to the tree. Since
    the tree is finitely branching, it must have an infinitely descending path.
    However, this infinitely descending path would correspond to an infinite descent
    in the underlying ordering, contradicting that hypothesis that $\ord_{X}$ is
    a WQO.
\end{proof}

\begin{lemma}\label{lemma:rpo_simp}
If $\opgeq$ is a total quasi-ordering, then \rpogeq is a quasi-simplification
ordering.
\end{lemma}
\begin{proof}

  We must show that $\rpogeq$ is a quasi-ordering, i.e it is reflexive and
  transitive; and also that it satisfies the replacement, subterm, and deletion
  properties.

  Reflexivity occurs via case 3 and \ref{lemma:mul_qo}.Replacement and deletion
  follow from case 3 of RPO and the definition of the multiset ordering.

    To prove the subterm property, we show a slightly stronger property: for all
  terms $t = f(\range{t}{m})$ and (not necessarily immediate) subterms
  $u = g(\range{u}{n})$, $t \rpogt u$. The proof goes by induction on the term
  size, where terms are bigger than their subterms, and by case analysis on the
  relationship between $f$ and $g$. Because $\opgeq$ is total, we have either
  $f \opgt g$, $f \approx g$, or $g \opgt f$.

    If $f \opgt g$, then to get $t \rpogeq u$ we must show
  $\{ t \} \rpogtm \srange{u}{n}$. Via induction, we have $t \rpogt u_{i}$ for
  all $1 \leqslant i \leqslant n$, as each $u_{i}$ is a subterm of $u$.
    To show $u \not\rpogeq t$, observe that we need
  $\srange{u}{n} \rpogeqm \{ t \}.$ This is impossible via the inductive
  hypothesis and the definition of $\rpogeqm$: we already have $t \rpogt u_{i}$
  for all $u_{i}$.

  If $f \approx g$, then we must show
  $\{\range{t}{m}\} \rpogtm \{\range{u}{n}\}$. If $u$ is a direct subterm of
  $t$, then $u = t_{i}$ for some $i$. By the inductive hypothesis we have
  $t_{i} \approx u \rpogt u_{j}$ for all $u_{j}$, which implies
  $\{\range{t}{m}\} \rpogtm \{\range{u}{n}\}$. If $u$ is a nested subterm, then
  we have some $t_{i} \rpogt u_{j}$ for all $u_{j}$ via the induction hypothesis:
  all $u_{j}$ are subterms of $t_{i}$.

  If $g \opgt f$, to get $t \rpogeq u$ then we must show $\srange{t}{m} \rpogeqm \{
u \}$. If $u$ was a direct subterm, then $t_{i} = u$ gives us the desired result;
otherwise we have $t_{i} \rpogt u$ via the inductive hypothesis.
  To show $u \not\rpogeq t$, observe that showing $u \rpogeq t$ would require
$\{ u \} \rpogtm \srange{t}{m}$. However we already have some $t_{i} \approx u$,
which prevents this possibility.

    Transitivity is also proven via induction on size. Assume we have
  $s = f(s_1, \ldots, s_{m}) \rpogeq{} t = g(t_{1}, \ldots, t_{n})$ and
  $t \rpogeq u = h(u_{1}, \ldots, u_{p})$. We proceed to show $s \rpogeq u$ by for
  each relationship between $f$, $g$, and $h$.
  \begin{enumerate}
    \item $f \opgt g \opgt h$, or $f \opgt g > h$: Via transitivity of
$\opgt$ we have $f \opgt h$, therefore we must show
$\{ s \} \rpogtm \{u_{1}, \ldots, u_{p}\}$. $\{ s \} \rpogeqm \{ t \}$ follows
from our assumption $s \rpogeq t$, and
$\{ t \} \rpogtm \{u_{1}, \ldots, u_{p} \}$ follows from $t \rpogeq u$. By the
inductive hypothesis, we have $s \rpogeq t \rpogeq u_{i}$ for all $u_{i}$, and
therefore $\{ s \} \rpogeqm \{ t \} \rpogtm \{u_{1}, \ldots, u_{p} \}$.
    \item $h \opgt g$: There must exist some subterm $t_{i}$ such that
          $t_{i} \rpogeq u$. Therefore we have
          $s \rpogeq t_{i}$ and $t_{i} \rpogeq u$, the inductive hypothesis
          gives us $s \rpogeq t_{i} \rpogeq u$.
    \item $g \opgt f$: There must exist some subterm $s_{i}$ such that
          $s_{i} \rpogeq t$. As above, using the induction hypothesis allows us
          to show $s_{i} \rpogeq u$, by the subterm property we have
          $s \rpogeq s_{i}$. We show $s \rpogeq u$ by the definition of $\rpogeq$.
    \item $f \approx g \approx h$. We clearly have $f \approx h$, we need to
          show $\srange{s}{m} \rpogeqm \srange{u}{p}$, which we have via \ref{lemma:mul_qo}.
  \end{enumerate}

\end{proof}

\begin{theorem}
  \label{theorem:rpo_wqo}
If~$\opgeq$ is a total \wqo, then \rpogeq{} is a \wqo.
\end{theorem}
\begin{proof}

    To show that $\rpogeq$ is \wqo, via the well-foundedness theorem of
  Dershowitz~\cite{dershowitz1982orderings}, which states that a
  quasi-simplification ordering $\geqslant'$ is \wqo if there exists a well-quasi
  ordering $\geqslant$ such that $f \geqslant g$ implies
  $f(\range{t}{n}) \geqslant' g(\range{t}{n})$.

  By \ref{lemma:rpo_simp} we have that $\rpogeq$ is a quasi-simplification
  ordering, and there exists an ordering over function symbols to satisfy the
  condition of the well-foundedness theorem: namely the underlying order
  $\opgeq$ from which $\rpogeq$ is constructed.

\end{proof}

\begin{theorem}
  If~$\opgeq$ is a total \wqo, then \rpogeq is thin
\end{theorem}
\begin{proof}
  We show that for any term $t = f(\range{t}{m})$, the set of terms
  $\{u ~|~ t \approx u = g(\range{u}{m}) \}$ is finite.

  If $t \approx u$, then we must have $t \rpogeq u$ and $u \rpogeq t$.
  Assume we have $t \rpogeq u$.

  First, we show that if $f > g$ then $u \not\rpogeq t$.
  Assume $u \rpogeq t$, then there must have some $u_{i}$ such that
  $u_{i} \rpogeq t$. But via the subterm property, we have
  $u \rpogt u_{i} \rpogeq t$, contradicting $t \rpogeq u$.

  Likewise, if $g > f$, then there is some $t_{i} \rpogeq {u}$. Then
  $t \rpogt t_{i} \rpogeq u$. Therefore we also have $u \not\rpogeq t$.

  Therefore, $t \approx u$ only if $f \approx g$. Since there are only a finite
  number of function symbols, then to show thinness we must show that only a
  finite number of multisets $\srange{u}{n}$ such that
  $\srange{t}{m} \rpogeqm \srange{u}{n}$ and $\srange{u}{n} \rpogeqm \srange{t}{m}$.
  If $\srange{t}{m} = \emptyset$, then the only such set is $\emptyset$.
  Otherwise, only such multisets are those where $\srange{u}{n}$
  is obtained from $\srange{t}{m}$ by removing zero or more terms $t_{i}$ and
  replacing them the same number of terms $u_{j}$ where $t_{i} \approx u_{j}$.
  If $\srange{t}{m} \rpogeqm \srange{u}{n}$ was justified by removing $t_{i}$
  from $\srange{t}{m}$ and removing smaller terms $\{u'~|~ u' < t_{i}\}$  from
  $\srange{u}{n}$, then we would have $\srange{t}{m} \rpogtm \srange{u}{n}$:
  this corresponds to the irreflexive single-step operation shown to form a
  well-founded order in lemma \ref{lemma:mul_wf}.

  Since the multisets contain a finite number of elements, and each term only
  has a finite number of equivalent terms (by induction on term size), there are
  only a finite number of such multisets.
\end{proof}

\clearpage
\section{Basic Equalities and Proved Theorems in the Program Equivalence Case Study}
\label{appendix:eval:casestudy}

\newcommand\addDist{1}
\newcommand\subDist{2}
\newcommand\timesTwoPlus{3}
\newcommand\plusZero{4}
\newcommand\mulZero{5}
\newcommand\mulOne{6}
\newcommand\subSelf{7}
\newcommand\divSelf{8}
\newcommand\subAdd{9}
\newcommand\mulSym{10}
\newcommand\addSym{11}
\newcommand\mulAssoc{12}
\newcommand\addAssoc{13}
\newcommand\ifT{14}
\newcommand\ifF{15}
\newcommand\seqNop{16}
\newcommand\seqNopP{17}
\newcommand\repeatNop{18}
\newcommand\repeatN{19}
\newcommand\ifJoin{20}
\newcommand\mapFusion{21}
\newcommand\foldMap{22}
\newcommand\foldFusion{23}

\begin{table}[H]
  \begin{tabular}{ |l|l|l| }
    \hline
    & Name & Formula  \\
  \hline
  1. & addDist & $(x * y) + (z * y) = (x + z) * y$
  \\
  2. & subDist & $(x * y) - (z * y) =  (x - z) * y$
  \\
  3. & times2Plus & $x * 2 = x + x$
  \\
  4. & plus0 & $x + 0 = x$
  \\
  5. & mul0& $x * 0 = 0$
  \\
  6. & mul1 & $x * 1 = x$
  \\
  7. & subSelf& $x - x = 0$
  \\
  8. & divSelf& $x / x = 1$
  \\
  \subAdd. & subAdd& $x - y = x + (-y)$
  \\
  \mulSym. & mulSym& $e * e' = e' * e$
  \\
  \addSym. & addSym& $e + e' = e' + e$
  \\
  \mulAssoc. & mulAssoc& $(x * y) * z = x * (y * z)$
  \\
  \addAssoc. & addAssoc& $(x + y) + z = x + (y + z)$
  \\
  \ifT. & ifT & \texttt{if True then lhs else rhs = lhs}
  \\
  \ifF. & ifF& \texttt{if False then lhs else rhs = rhs}
  \\
  \seqNop . & seqNop & \texttt{seq lhs nop = lhs}
  \\
  \seqNopP. & seqNop\textquotesingle & \texttt{seq nop rhs = rhs}
  \\
  \repeatNop . & repeatNop & \texttt{repeat 0 body = nop}
  \\
  \repeatN . & repeatN1 & \texttt{repeat (S n) body = seq body (repeat n body)}
  \\
  \ifJoin . & ifJoin & \texttt{if c1 then (if c2 then op else nop) else nop}
  \\
  & & \texttt{ = if (c1 and c2) then op else nop}
  \\
  \mapFusion. & mapFusion &  \texttt{map g (map f xs) = map (g . f) xs}
  \\
  \foldMap. & foldMap & \texttt{(foldr f e) . (map g) = foldr (f . g) e}
  \\
  \foldFusion. & foldFusion & $\forall$ \texttt{x y . h (f x y) = f\textquotesingle~x (h y)}
  \\
  & &            $\implies$ \texttt{h . (foldr f e) xs = foldr f\textquotesingle~(h e) xs}
  \\
  \hline
  \end{tabular}
  \caption{Basic Equality Axioms used in our Program Equivalence Case Study }\label{table:ast-basic-eq}
  \end{table}
  \begin{table}
\begin{tabular}{ |l|l| }
  \hline
Formula	& Rewrites\\
\hline
$(-(x + x)) + (x + x) = 0$ & 	\subSelf, \addSym, \subAdd
\\
$(x * 2) * 2 = (x + x + x + x)$	& \timesTwoPlus, \addAssoc
\\
$(x * y) + (y * x) = (x * 2 * y)$	& \timesTwoPlus, \mulSym, \mulAssoc
\\
$(x * y) + (y * z) - ( (x + z) * y) = 0$	& \addDist, \subSelf, \mulSym
\\
$(x * y) - (0 * y) = x * y$	& \subDist, \subAdd, \subSelf, \plusZero
\\
$x * (1 - (x / x)) = 0$	& \mulZero, \subSelf, \divSelf
\\
$x * 1 = x + 0$	& \plusZero, \mulOne
\\
$\texttt{if true then (seq nop hw) else nop = hw}$ & 	\seqNopP, \ifT
\\
$\texttt{repeat (S (S Z)) hw = seq hw hw}$ & \seqNop, \repeatNop, \repeatN
\\
$\texttt{if True then (if False then hw else nop) else nop}$ & \\
$\texttt{ = if (True and False) then hw else nop}$ &	\ifJoin
\\
$\texttt{map p1 (map p2 list) = map p3 list}$ &	\mapFusion
\\
$\texttt{( (foldr add 0) . (map p1)) list = foldr addP1 0 list}$ &	\foldMap
\\
$\texttt{double . (foldr add 0) list = foldr twicePlus 0 list}$ &	\foldFusion
\\
\hline
\end{tabular}
\caption{\small{Theorems Proved via Rewriting using the Basic Equality axioms in
  \ref{table:ast-basic-eq}}}
\vspace{-0.2in}
\label{table:ast-checked-eq}
\end{table}

\clearpage
\section{Dafny Matching Loop Example}
\label{appendix:evaluation:dafnyloop}
\begin{figure}[H]
  \begin{mcode}
datatype List = Nil | Cons(head: int, tail: List)

function append(xs: List, ys: List): List
{
	match xs
	case Nil => ys
	case Cons(x, xrest) => Cons(x, append(xrest, ys))
}

function reverse(xs: List): List
{
	match xs
	case Nil => Nil
	case Cons(x, xrest) => append(reverse(xrest), Cons(x, Nil))
}

lemma AppendNil(xs: List)
    ensures append(xs, Nil) == xs; {}

lemma AppendAssoc(xs: List, ys: List, zs: List)
  ensures append(xs, append(ys, zs)) == append(append(xs, ys), zs); {}

lemma ReverseAppendDistrib(xs: List, ys: List)
  ensures reverse(append(xs, ys)) == append(reverse(ys), reverse(xs));
{
    forall xs : List {AppendNil(xs);}
    forall xs : List, ys: List, zs: List {AppendAssoc(xs, ys, zs);}
}

lemma ReverseInvolution(xs: List)
    ensures reverse(reverse(xs)) == xs;
{
    // Axiom definition inserted here
    { forall (xs, ys) { ReverseAppendDistrib(xs, ys); } }

    match xs {
      case Nil =>
      case Cons(x, xrest) =>
        calc { // Equational reasoning steps removed here
            reverse(reverse(xs));
            {ReverseInvolution(xrest);}
            xs;
         }
    }
}
  \end{mcode}
  \caption{A version of the reverse involution proof (\figref{eval:involution})
from \cite{leino2013verified} with intermediate equality steps removed.
Attempting to verify this code causes a matching loop when using Dafny version
3.3.0 and Z3 version 4.8.5}
\end{figure}

\end{document}